\documentclass[12pt]{article}

\usepackage{stmaryrd}
\usepackage{amsmath}
\usepackage{graphicx}
\usepackage{color}
\usepackage[title]{appendix}
\usepackage{amsfonts,amscd,amssymb, mathtools,mathrsfs, dsfont}

\usepackage{pgfplots}
\usepackage{float}

\usepackage{url}
\usepackage{hyperref}

\usepackage{array}

\usepackage[ vmargin=2cm,hmargin=2cm]{geometry} 
\newif\ifblog
\newif\iftex
\blogfalse
\textrue

\newcommand{\thmref}[1]{Theorem~{\rm \ref{#1}}}
\newcommand{\lemref}[1]{Lemma~{\rm \ref{#1}}}
\newcommand{\corref}[1]{Corollary~{\rm \ref{#1}}}
\newcommand{\propref}[1]{Proposition~{\rm \ref{#1}}}
\newcommand{\defref}[1]{Definition~{\rm \ref{#1}}}

\newcommand{\be}{\begin{eqnarray}}
\newcommand{\ee}{\end{eqnarray}}
\newcommand{\bee}{\begin{eqnarray*}}
\newcommand{\eee}{\end{eqnarray*}}

\newcommand{\cc}{\overline{c}}

\def\wv{\widehat{v}}
\def\wu{\widehat{u}}

\def\d{\:{\rm d}}
\def\dc{\:{\rm d}c}
\def\dt{\:{\rm d}t}
\def\ds{\:{\rm d}s}

\def\dy{\:{\rm d}y}
\def\P{{\mathbb P}}
\def\E{{\mathbb E}}

\def\R{{\mathbb R}}

\def\Q{{\cal Q}}

\def\p{{\partial}}

\newcommand{\cF}{{\cal F}}
\newcommand{\ep}{\varepsilon}
\newcommand{\al}{\alpha}

\newcommand{\nd}{\noindent}

\newcommand{\ga}{\gamma}
\newcommand{\si}{\sigma}

\newcommand{\Om}{\Q_{0}}

\newcommand{\esssup}{\mathop{\rm ess\;sup}\limits}

\newtheorem{theorem}{Theorem}[section]
\newtheorem{lemma}[theorem]{Lemma}
\newtheorem{definition}[theorem]{Definition}
\newtheorem{corollary}[theorem]{Corollary}
\newtheorem{proposition}[theorem]{Proposition}

\newenvironment{proof}{\noindent{\sc Proof:}}{\strut\hfill $\Box$\medskip\\} 

\newcommand{\LL}{{\cal L}}

\def\wv{\widehat{v}}
\def\wu{\widehat{u}}
\def\wx{\widehat{x}}

\def\ou{\overline{u}}
\def\ov{\overline{v}}

\def\Dc{\Delta c}
\def\A{{\cal A}}
\def\B{{\cal B}}
\def\setb{{\cal C}}

\def\calE{{\cal E}}

\def\sw{\mathcal{X}}
\def\swi{\mathfrak{C}_{[c,\cc]}}

\def\Wploc{W^{2,\:p}_{\rm loc}}
\def\Wp{W^{2, p}}

\def\seta{C^{\infty,0}(\Om)}
\def\setb{\mathcal{D}}

\def\DD{\mathcal{C}}

\selectfont

\title{Optimal ratcheting of dividend payout under Brownian motion surplus}

\author{Chonghu Guan\thanks{School of Mathematics, Jiaying University, Meizhou 514015, Guangdong, China. Email: \url{gchonghu@163.com}.}
\and Zuo Quan Xu\thanks{Department of Applied Mathematics, The Hong Kong Polytechnic University, Kowloon, Hong Kong, China. Email: \url{maxu@polyu.edu.hk}}
}
\date{}

\begin{document}
\maketitle
\vspace{-20pt}
\begin{abstract} 
This paper is concerned with a long standing optimal dividend payout problem subject to the so-called ratcheting constraint, that is, the dividend payout rate shall be non-decreasing over time and is thus self-path-dependent. The surplus process is modeled by a drifted Brownian motion process and the aim is to find the optimal dividend ratcheting strategy to maximize the expectation of the total discounted dividend payouts until the ruin time.
Due to the self-path-dependent control constraint, the standard control theory cannot be directly applied to tackle the problem. The related Hamilton-Jacobi-Bellman (HJB) equation
is a new type of variational inequality. In the literature, it is only shown to have a viscosity solution, which is not strong enough to guarantee the existence of an optimal dividend ratcheting strategy. This paper proposes a novel partial differential equation method to study the HJB equation.
We not only prove the the existence and uniqueness of the solution in some stronger functional space, but also prove the strict monotonicity, boundedness, and $C^\infty$-smoothness of the dividend ratcheting free boundary. Based on these results, we eventually derive an optimal dividend ratcheting strategy, and thus solve the open problem completely.
Economically speaking, we find that if the surplus volatility is above an explicit threshold, then one should pay dividends at the maximum rate, regardless the surplus level. Otherwise, by contrast, the optimal dividend ratcheting strategy relays on the surplus level and one should only ratchet up the dividend payout rate when the surplus level touches the dividend ratcheting free boundary. {Moreover, our numerical results suggest that one should invest into those companies with stable dividend payout strategies since their income rates should be higher and volatility rates smaller.}
\bigskip\\
\nd {\bf Keywords.} Free boundary; variational inequity; self-path-dependent constraint.
\bigskip\\
\nd {\bf 2010 Mathematics Subject Classification.} 
35R35; 35Q93; 91G10; 91G30; 93E20.

\end{abstract}
 
 \section{Introduction}
The study of optimally paying dividends from a dynamic stochastic surplus process goes back at least to De Finetti \cite{de1957impostazione} and Gerber \cite{gerber1969entscheidungskriterien}.
In an optimal dividend payout problem, the objective (of the company) is to find an optimal dividend payout strategy to maximize the expectation of the total discounted dividend payouts until bankruptcy (i.e. the ruin time).
The optimal dividend payout strategy shall be a tradeoff between the dividend compensations to the shareholders and the managed surplus process to secure the position so as to avoid or delay the ruin time. The underlying surplus process can be modeled in many ways. The most popular ones include compound Poisson model (\cite{gerber2006optimal}, \cite{albrecher2020optimal}), drifted Brownian motion model (\cite{asmussen1997controlled}, \cite{gerber2004optimal}, \cite{azcue2005optimal}), jump-diffusion model (\cite{belhaj2010optimal}), {L\'evy model (\cite{Kyprianou2012},\cite{BayraktarKyprianou2013}), and other diffusion models (\cite{BayraktarEgami2010}, \cite{reppen2020optimal})}. The dividend payout strategy can be constrained to different types as well; for instance, {
Bayraktar and Young \cite{BayraktarYoung2008} considered strategies which are functions of historical maximum wealth; Bayraktar and Egami \cite{BayraktarEgami2010} investigated impulse strategies; Angoshtari et al. \cite{angoshtari2019optimal} studied excess dividend rate with a drawdown constraint}; 
see \cite{avanzi2009strategies} for an overview).

In this paper, we focus on the long standing optimal dividend payout problem, where the surplus process is modeled by a drifted Brownian motion, and the dividend payout is subject to the so-called \textit{dividend ratcheting} constraint, that is, the dividend payout rate shall be non-decreasing over time.
The ratcheting (namely, non-decreasing) constraint is a particular type of habit-formation that
has been extensively investigated in financial economic literature. Dybvig \cite{dybvig1995dusenberry} first considered a life-time portfolio selection model with consumption ratcheting where no decline is allowed in the consumption rate of the agent. Similar problems with ratcheting over the wealth has been studied by Roche \cite{roche2006optimal}, Elie and Touzi \cite{elie2008optimal} and Chen et al. \cite{CLLL15}.
Concerning the dividend payout problem,
Albrecher \cite{albrecher2018dividends} considered a two-level ratcheting constraint problem, that is, one can only ratchet up once from a lower level dividend rate to a higher one. 
Albrecher et al. \cite{albrecher2020optimal} and \cite{albrecher2022optimal} investigated the optimal dividend problems with drawdown and ratcheting constraint under different surplus models by viscosity solution theory.

Technically, with ratcheting constraint involved, the optimal dividend payout problems shall lead to optimal control problems with self-path-dependent control constraint. To the best of our knowledge, there seems no universal way to deal with such type of stochastic control problems. The related Hamilton-Jacobi-Bellman (HJB) equations for this new type of problems are variational inequalities with at least two arguments: a state argument (representing the surplus level) and a control argument (representing the historical maximum dividend payout rate).
Unlike those classical variational inequalities with function constraint (such as Black-Scholes partial differential equations for American options), the HJB equations form a new type of variational inequalities where a gradient constraint on the value function against the dividend payout rate is involved.
Although similar HJB equations with gradient constraints have appeared in the problems involving transaction costs (see \cite{DY09,DXZ10}), they are different in nature. In transaction problems, the gradient constraints are put on the state argument, whereas in dividend ratcheting problems, they are on the control argument.
They are so different such that cannot be treated by similar methods. 
Albrecher et al. \cite{albrecher2022optimal} considered an optimal dividend payout problem with ratcheting constraint. They showed that the value function is the unique viscosity solution to corresponding HJB equation. Also, the value function can be approximated by problems with finite dividend ratcheting. But, as is well-known, viscosity solution is so weak such that they cannot provide an optimal dividend payout strategy for the original problem.

In this paper, we study the same problem as in \cite{albrecher2022optimal}. Different from the existing literature that majorly use viscosity solution technique to study the HJB equations, we propose a novel partial differential equation (PDE) method to study the HJB equation.
Following the similar idea of discretization in \cite{albrecher2022optimal}, we first disperse the parameter to obtain a sequence of ordinary differential equations (ODEs) whose solvability is well-known.
By establishing various estimates and taking limit, we can get a fairly strong solution to the HJB equation. We next define a dividend ratcheting free boundary and derive its properties such as boundedness, monotonicity, and $C^{\infty}$-smoothness.
Finally, we can use this dividend ratcheting free boundary to construct a complete answer to the original optimal dividend payout problem under ratcheting constraint. 
We also perform a numerical analysis to exam the effects of different parameters (including the maximum allowable dividend payout rate, the discount rate, the income rate and the volatility rate) on the value function and the free boundary. The results suggest that one should invest
into those companies with stable dividend payout strategies since their income rates should be higher and volatility rates smaller.

The remainder of this paper is organized as follows. We first introduce our optimal dividend ratcheting problem in Section \ref{section2}. The boundary case and a simple case will be solved in this section as well. We introduce the HJB equation for the complicated case and solve the problem completely in Section \ref{section3}. A numerical analysis is presented in Section \ref{section4} to exam the effects of different parameters on the value function and the free boundary. 
Section \ref{section5} is devoted to the solvability of the HJB equation and Section \ref{section6} devoted to the study of the dividend ratcheting free boundary.


\section{Optimal Dividend Ratcheting Problem} \label{section2}

We use a filtered complete probability space $(\Omega, \cF, \P, \{\cF_t\}_{t\geq0})$ to
represent the financial market, where the filtration $\{\cF_t\}_{t\geq0}$ is generated by a standard one-dimensional Brownian motion $\{W_t, t\geq 0\}$ defined in the probability space, argumented with all $\P$ null sets.

We assume the income of a company follows a drifted Brownian motion process. After paying dividends, the surplus process of the company follows the following stochastic differential equation (SDE):
\begin{equation}\label{X_eq}
\d X_t=(\mu-\DD_t) \dt+\sigma \d W_t,~\; t\geq 0,
\end{equation}
where the constant $\mu>0$ denotes the income rate of the company, and the constant $\sigma>0$ represents the volatility rate of the surplus process, $\DD_t$ is the dividend payout rate at time $t$ to be chosen by the company.

Now let us introduce the set of admissible dividend payout strategies, which is the distinguish feature of our model. Economically speaking, in order to survive, a rational company should not pay dividends at a rate higher than its income, so we fix a maximum dividend payout rate $\cc$ throughout the paper such that
$$\cc\in[0, \mu].$$
Given an initial (historical maximum) dividend payout rate $c\in[0,\cc]$, we call a dividend payout strategy $\{\DD_t\}_{t\geq0}$ an admissible
{\it dividend ratcheting strategy}
if it is an $\{\cF_t\}_{t\geq0}$ adapted, non-decreasing process such that $c\leq \DD_t\leq \cc$ for all $t\geq 0$.
The set of these ratcheting strategies is denoted by $\Pi_{[c,\cc]}$.
Note the value of a ratcheting strategy chosen at the current time will affect all the future choices of it since the future rates cannot be lower than the current one. The higher the current choice, the less choices the future, so ratcheting strategies are self-path-dependent. This is a critical difference between our problem and those standard stochastic control problems in Yong and Zhou \cite{YZ99} where controls do not relay on its historical choices. 

The company's objective is to find an admissible dividend ratcheting strategy $\{\DD_t\}_{t\geq0}$ to maximize the total discounted dividend payouts until the ruin time.
Mathematically, we need to solve
\begin{align}\label{value}
V(x,c)=
\sup\limits_{\{\DD_t\}_{t\geq0}\in \Pi_{[c,\cc]}}\E \Bigg[\int_0^\tau e^{-r t} \DD_t \dt \;\bigg|\; X_0=x\Bigg],~ (x,c)\in\Q,
\end{align}
where the surplus process $X_{t}$ follows the SDE \eqref{X_eq} with an initial value $X_{0}=x\in \R^{+}$, $r>0$ is a constant discount rate,
$\tau$ is the ruin time defined by
$$\tau:=\inf\big\{t\geq 0 \;\big|\; X_t\leq 0\big\}$$ and 
$$\Q:=\R^+\times [0,\cc],~ \R^{+}:=[0,\infty).$$
Note that the ruin time $\tau$ does depend on the dividend ratcheting strategy $\{\DD_t\}_{t\geq0}$. We will not emphasis this point in the future. {Because the surplus process $X$ in \eqref{X_eq} is continuous, we always have $X_{\tau}=0$. If the surplus process were modeled with downward jumps, one may not always have $X_{\tau}=0$. The mathematical treatement would be more involved than our model. }

Since $0\leq \tau\leq \infty$ and $0\leq \DD_t\leq \cc$ for any admissible strategy $\{\DD_t\}_{t\geq0}$, it yields 
\begin{align}\label{valueestimate}
0\leq V(x,c)\leq\int_0^{\infty} e^{-r t} \cc \dt=\frac{\cc}{r}.
\end{align}
It is also easy to observe that $V(x,c)$ is non-decreasing in $x$ and non-increasing in $c$.

We emphasis again, the problem \eqref{value} is an optimal stochastic control problem with self-path-dependent control constraint, that is, the value of $\DD_t$ will affect the choice of $\DD_s$ for any time $s>t$. Unlike those standard stochastic control problems studied in \cite{YZ99},
there seems no uniform way to deal with such type of stochastic control problems.
This paper proposes a PDE method to tackle the problem \eqref{value}. To this end, we need first to give its HJB equation. Since the boundary value is required for the HJB equation, let us start with the boundary case.

\subsection{The Boundary Case: $V(x,\cc)$}
In the boundary case, i.e., the initial dividend payout rate $c$ is equal to the maximum rate $\cc$, so there is only one admissible (and thus optimal) strategy in $\Pi_{[\cc,\cc]}$, namely $\DD_t\equiv\cc$. It yields 
\begin{align*} 
V(x,\cc)=\E \Bigg[\int_0^\tau e^{-r t} \cc\;\dt \;\Bigg|\; X_0=x \Bigg]
=\frac{\cc}{r} \Big(1-\E\big[e^{-r \tau} \big]\Big),~ x\in\R^{+},
\end{align*}
where the ruin time becomes
$$\tau=\inf\big\{t\geq 0 \;\big|\; x+(\mu-\cc)t+\sigma W_t\leq 0\big\}.$$

\begin{lemma} It holds that
$V(x,\cc)=g(x),~ x\in\R^{+},$ where
\begin{align}\label{def:g}
g(x):=\frac{\cc}{r}(1-e^{-\ga x})\geq 0,~ x\in\R^{+},
\end{align}
and
$\ga:=\frac{(\mu-\cc)+\sqrt{(\mu-\cc)^2+2 \si^2 r}}{\si^2}>0$
is the positive root of
$$-\frac{1}{2}\si^2 \ga^2+(\mu-\cc) \ga+r=0.$$
\end{lemma} 
\begin{proof} 
Let $M_{t}=e^{-\frac{1}{2}\ga^{2}\sigma^{2}t-\ga\sigma W_{t}}$. Then it is a martingale. If $\tau<\infty$, we then have $x+(\mu-\cc)\tau+\sigma W_{\tau}=0,$ so that, for any constant $T>0$, 
$$M_{\tau}\mathbf{1}_{\tau\leq T}=e^{-\frac{1}{2}\ga^{2}\sigma^{2}\tau-\ga\sigma W_{\tau}}\mathbf{1}_{\tau\leq T}=e^{-\frac{1}{2}\ga^{2}\sigma^{2}\tau+\ga(\mu-\cc)\tau+\ga x}=e^{-r\tau+\ga x}\mathbf{1}_{\tau\leq T}.$$ 
Hence, by Doob's optional stopping theorem,
\begin{align}\label{doob}
\E\big[e^{-r\tau+\ga x}\mathbf{1}_{\tau\leq T}\big]+\E\big[M_{T}\mathbf{1}_{\tau> T}\big]=\E\big[M_{\tau}\mathbf{1}_{\tau\leq T}+M_{T}\mathbf{1}_{\tau> T}\big]=\E\big[M_{\tau\wedge T}\big]=1.
\end{align} 
When $\tau>T$, we have $x+(\mu-\cc)T+\sigma W_{T}>0$, so
$$0\leq M_{T}\mathbf{1}_{\tau> T}\leq e^{-\frac{1}{2}\ga^{2}\sigma^{2}T+\ga(\mu-\cc)T+\ga x}\mathbf{1}_{\tau> T}=e^{-rT+\ga x}\mathbf{1}_{\tau> T}\leq e^{-rT+\ga x}.$$
Sending $T\to\infty$ in \eqref{doob} and using the dominated convergence theorem, we get 
$\E\big[e^{-r\tau+\ga x}\mathbf{1}_{\tau< \infty}\big]=1$, so $\E\big[e^{-r\tau}\big]=e^{-\ga x}$. This leads to $V(x,\cc)=g(x)$.
\end{proof}

We remark that $g(\cdot)$ is the unique power growth (indeed bounded) solution to the following ODE with Dirichlet boundary condition:
\begin{align}\label{g_eq}
\begin{cases}
-\LL_{\cc}g-\cc=0,~ x\in\R^{+}, \medskip\\
g(0)=0,
\end{cases}
\end{align}
where the operator $\LL_c$ is defined as
\begin{align}\label{defop}
\LL_c v:=\frac{1}{2}\si^2v_{xx}+(\mu-c)v_x-r v.
\end{align}
We use the function $g$ and operator $\LL_c$ defined above throughout this paper.

Figure \ref{fig:b1} illustrates the function $g$ when $\mu=0.4$, $r=0.05$, $\sigma=0.4$, $\cc=0.3$:
\begin{figure}[H]
\centering
\includegraphics[width=0.5\linewidth]{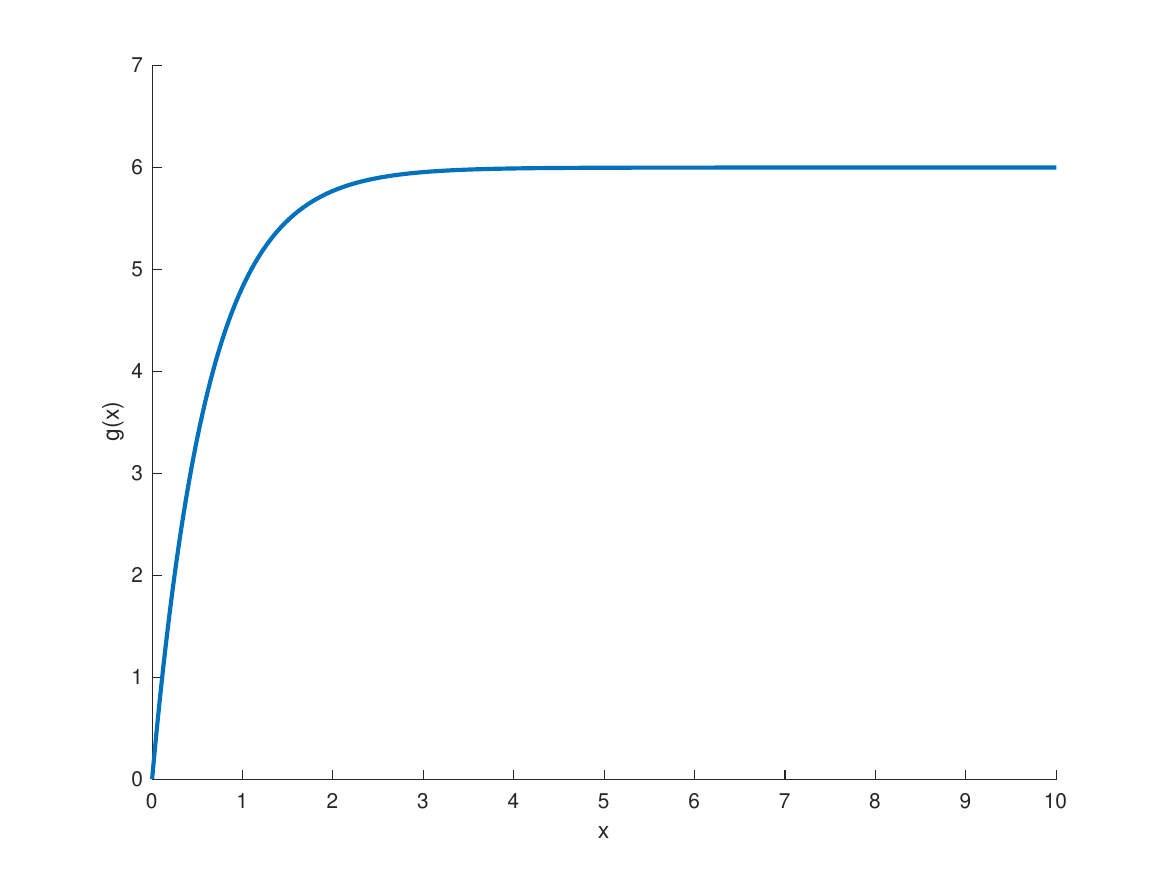}
\caption{The function $g(x)$ when $\mu=0.4$, $r=0.05$, $\sigma=0.4$, $\cc=0.3$}
\label{fig:b1}
\end{figure}


Before giving the HJB equation, we first study a simple case, for which we can give an explicit optimal dividend ratcheting strategy and optimal value.

\subsection{The Simple Case: $2\mu\cc\leq \si^2r.$}

\begin{theorem}\label{simplecase}
Suppose $2\mu\cc\leq \si^2r$. Then $V(x,c)=g(x)$ for all $(x,c)\in\Q$ and $\DD^{*}_t\equiv\cc$
is an optimal dividend ratcheting strategy to the problem \eqref{value}.
\end{theorem}
To prove this theorem, we need the following technical result.
\begin{lemma}\label{simpleineq}
The inequality $\cc \ga \leq r$ holds if and only if $2\mu\cc\leq \si^2r$.
\end{lemma} 
\begin{proof}
Note that the function $f(y)=-\frac{1}{2}\si^2 y^2 +(\mu-\cc) y+r$ satisfies $f(y)> 0$ if $y\in (0,\ga)$ and $f(y)\leq 0$ if $y\geq \ga$, so $ \ga \leq r/\cc$ is equivalent to $f(r/\cc)\leq 0$, i.e., $2\mu\cc\leq \si^2r$.
\end{proof}

\noindent{\sc Proof of Theorem \ref{simplecase}:}
If $2\mu\cc\leq \si^2r$, then by Lemma \ref{simpleineq}, $$g'(x)=\frac{\cc \ga}{r}e^{-\ga x}\leq 1,~ x\in\R^{+}.$$
This together with \eqref{g_eq} yields
\begin{align*}
-\LL_c g-c 
=-\LL_c g-c-(-\LL_{\cc} g-\cc)=(c-\cc)(g'-1)\geq 0,~ x\in\R^{+}.
\end{align*}
Therefore, for any constant $T>0$ and admissible strategy $\{\DD_t\}_{t\geq 0}\in \Pi_{[c,\cc]}$, applying It\^o's formula to $e^{-rt}g(X_{t})$ on $[0,\tau\wedge T]$ gives
\begin{align*}
g(x)
&=\E [e^{-r(\tau\wedge T) }g(X_{\tau\wedge T})]
-\E \Bigg[\int_0 ^{\tau\wedge T }e^{-r t } \LL_{\DD_t}g(X_t) \dt\Bigg]\\
&\quad\;-\E \Bigg[\int_0 ^{\tau\wedge T }e^{-r t }\si g'(X_t) \d W_t\Bigg] \geq \E \Bigg[\int_0 ^{\tau\wedge T }e^{-r t } \ \DD_t \dt\Bigg],
\end{align*}
where we used the fact that $g\geq 0$ and $g'$ is bounded to get the inequality.
Sending $T\to\infty$ in above, we get from the monotone convergence theorem that
\begin{align}\label{gbiggerthan}
g(x) &\geq \E \Bigg[\int_0 ^{\tau }e^{-r t } \ \DD_t \dt\Bigg].
\end{align}
Since $\{\DD_t\}_{t\geq 0}\in \Pi_{[c,\cc]}$ is arbitrary selected, the above shows that $g\geq v$ on $\Q$.

On the other hand,
under the special dividend ratcheting strategy $\DD^{*}_t\equiv\cc$, which is clearly admissible, one can show similarly to the boundary case that 
\begin{align*}
g(x)&= \E \Bigg[\int_0 ^{\tau }e^{-r t } \ \DD^{*}_t \dt\Bigg].
\end{align*}
This together with $g\geq v$ on $\Q$ shows that $\DD^{*}_t\equiv\cc$
is an optimal dividend ratcheting strategy to the problem \eqref{value} and $v=g$ on $\Q$.
\strut\hfill $\Box$\\

From this result we can see that if the maximum dividend payout rate is relatively small (i.e. $\cc\leq \si^2r/(2\mu))$, then it is optimal for the company to pay dividends to the shareholders at the maximum rate all the time, regardless its surplus level. Intuitively speaking, a higher dividend payout rate benefits the utility more than its negative impact on the survival time.
We can also interpret the result from another point of view. If the income process has a relatively high risk/uncertainty (i.e. $2\mu\cc/r\leq \si^2$), then it is optimal to pay dividends at the maximum rate. Intuitively speaking, because of the highly uncertainty of the income process, the company has a higher risk to be bankrupt in the near future, so it is better to pay dividends as soon as possible.

In the rest of the paper, we deal with the more involved case: $2\mu\cc>\si^2r$, which is assumed from now on. Note the condition $2\mu\cc>\si^2r$ by Lemma \ref{simpleineq} is equivalent to that $g'(0)>1$.

\section{The HJB equation and optimal dividend ratcheting strategy in the complicated case $2\mu\cc>\si^2r$}\label{section3}

The HJB equation for the optimization problem \eqref{value} is a variational inequality on $v:\Q\to \R$:
\begin{align}\label{v_pb}
\begin{cases}
\min\big\{\!-\LL_c v- c,~-v_c\big\}=0, & (x,c)\in\Q,\medskip\\
v(0,c)=0, & c\in [0,\cc],\medskip\\
v(x,\cc)=g (x), & x\in \R^+,
\end{cases}
\end{align}
where
the function $g$ is defined in \eqref{def:g}, and the operator $\LL_c$ is defined in \eqref{defop}.
Thanks to the estimate \eqref{valueestimate}, it suffices to study bounded solution(s) to the above HJB equation \eqref{v_pb}. Note that $v=g$ is a solution to \eqref{v_pb} if $2\mu\cc\leq\si^2r$, and not a solution otherwise.

It is proved in \cite{albrecher2020optimal} that the HJB equation \eqref{v_pb} admits a unique viscosity solution, which is the value function $V$ defined in \eqref{value}.
In this paper, we will further prove that \eqref{v_pb} admits a solution in the following stronger sense. Unlike \cite{albrecher2020optimal}, our solution will allow us to construct an optimal dividend ratcheting strategy for the problem \eqref{value}, thus solve the problem completely.
\begin{definition}\label{hjb1}
A function $v:\Q\to\R$ is called a solution to the HJB equation \eqref{v_pb} if
\begin{enumerate}
\item it holds that $v\in\A$, where
\begin{align*}
\A=\Big\{v:\Q\to \R \;\Big|\; & \mbox{$v$, $v_{x}$ and $v_{c}$ are continuous and bounded }\\
&\mbox{in $\Q$; $v_{c}\leq 0$ in $\Q$; and $v(\cdot,c)\in \Wploc(\R^+)$}\\
&\mbox{for each $c\in [0,\cc]$ and each $p>1$}\Big\};
\end{align*}
\item it satisfies the boundary conditions in \eqref{v_pb};
\item it holds that
\begin{align}\label{-Lv>=0}
-\LL_c v(\cdot,c)-c \geq 0 \hbox{\; a.e.\;in $\R^+$, for each $c\in[0,\cc]$};
\end{align}

\item if $v_{c}(x,c)<0$, then
\begin{align}\label{-Lv=0}
-\LL_c v(y,c)-c=0 \hbox{\; for all $y\in[0,x]$}.
\end{align}

\end{enumerate}
\end{definition}
As usual, the notation $\Wploc$ stands for the Sobolev space.
We remark that the last requirement \eqref{-Lv=0} is in the classical sense rather than the strong sense. 

In the variational inequality in \eqref{v_pb}, the obstacle is against $v_c$, but the equation is for $v$, so it is natural to transform the variational inequality \eqref{v_pb} into someone for $u=-v_c$ which together with $v(x,\cc)=g(x)$ implies 
$$v(x,c)=g (x)+\int_c^{\cc} u(x,s)\ds.$$
It thus follows $$\p_c(-\LL_c v - c)=\LL_c (-v_{c})+v_x-1=\LL_c u+g'(x)+\int_c^{\cc} u_x(x,s) \ds-1.$$
Therefore, we conjecture an obstacle problem for $u$ as follows:
\begin{align}\label{u_pb}
\begin{cases}
\displaystyle\min\bigg\{\!-\LL_c u-g'(x)-\int_c^{\cc} u_x(x,s) \ds+1,~u\bigg\}=0, & (x,c)\in\Q,\medskip\\
u(0,c)=0, & c\in [0,\cc].
\end{cases}
\end{align}
This is a single-obstacle problem with a nonlocal operator.
\begin{definition}\label{hjb2}
We call a function $u:\Q\to\R^{+}$ is a solution to the variational inequality \eqref{u_pb} if
\begin{enumerate}
\item it holds that $u\in\B$, where
\begin{align*}
\B=\Big\{u:\Q\to \R^{+} \;\Big|\; & \mbox{both $u$ and $u_{x}$ are continuous and bounded in $\Q$; and} \\
&\mbox{$u(\cdot,c)\in \Wploc(\R^+)$ for each $c\in [0,\cc]$ and each $p>1$}\Big\};
\end{align*}
\item for each $c\in[0,\cc]$, it holds that
\begin{align}\label{-Lu>=0}
-\LL_c u-g'(x)-\int_c^{\cc} u_x(x,s) \ds+1\geq 0 \hbox{\; a.e.\;in $\R^+$};
\end{align}

\item if $u(x,c)>0$ for some $(x,c)\in\Q$, then it holds that
\begin{align}\label{-Lu=0}
-\LL_c u(y,c)-g'(y)-\int_c^{\cc} u_x(y,s) \ds+1=0 \hbox{\; for all $y\in[0,x]$}.
\end{align}
\end{enumerate}
\end{definition}

Our next main result concerns the solvability of \eqref{v_pb} and \eqref{u_pb}. 
\begin{theorem}\label{thm:u}
There exists a pair $(v,u)$ that satisfies the relation 
\begin{align}\label{v_u}
v(x,c)=g (x)+\int_c^{\cc} u(x,s)\ds, ~(x,c)\in\Q.
\end{align}
Also, $v$ and $u$ are, respectively, solutions to \eqref{v_pb} and \eqref{u_pb} in the sense of \defref{hjb1} and \defref{hjb2}.
Furthermore,
for each constant $\alpha\in(0,1)$, there is a constant $K>0$ such that, for all $(x,c)\in\Q$,
the following estimates hold: 
\begin{gather}
\label{v} 0\leq v\leq \frac{\cc}{r},\\
\label{vx} 0\leq v_x\leq K,\\
\label{vc} 0\leq -v_c\leq K,\\ 
\label{vxx} v_x(y,c)\leq \max\{v_x(x,c), 1\},~\forall\; 0\leq x\leq y,\\
\label{ub}|u(\cdot,c)|_{C^{1+\al}(\R^+)}\leq K,\\
\label{ucb} |u_c(\cdot,c)|\leq K.
\end{gather}
\end{theorem}
This result will be proved in Section \ref{sec:solvability}.

We can also see from \eqref{vxx} that the map $x\to v_x(x,c)$ is concave in the region $\{v_{x}>1\}$.
Indeed, the numerical result in Figure \ref{fig:b1v} indicates that it is concave for all $x\in\R^{+}$. However, we cannot prove this.
\begin{figure}[H]
\centering
\includegraphics[width=0.8\linewidth]{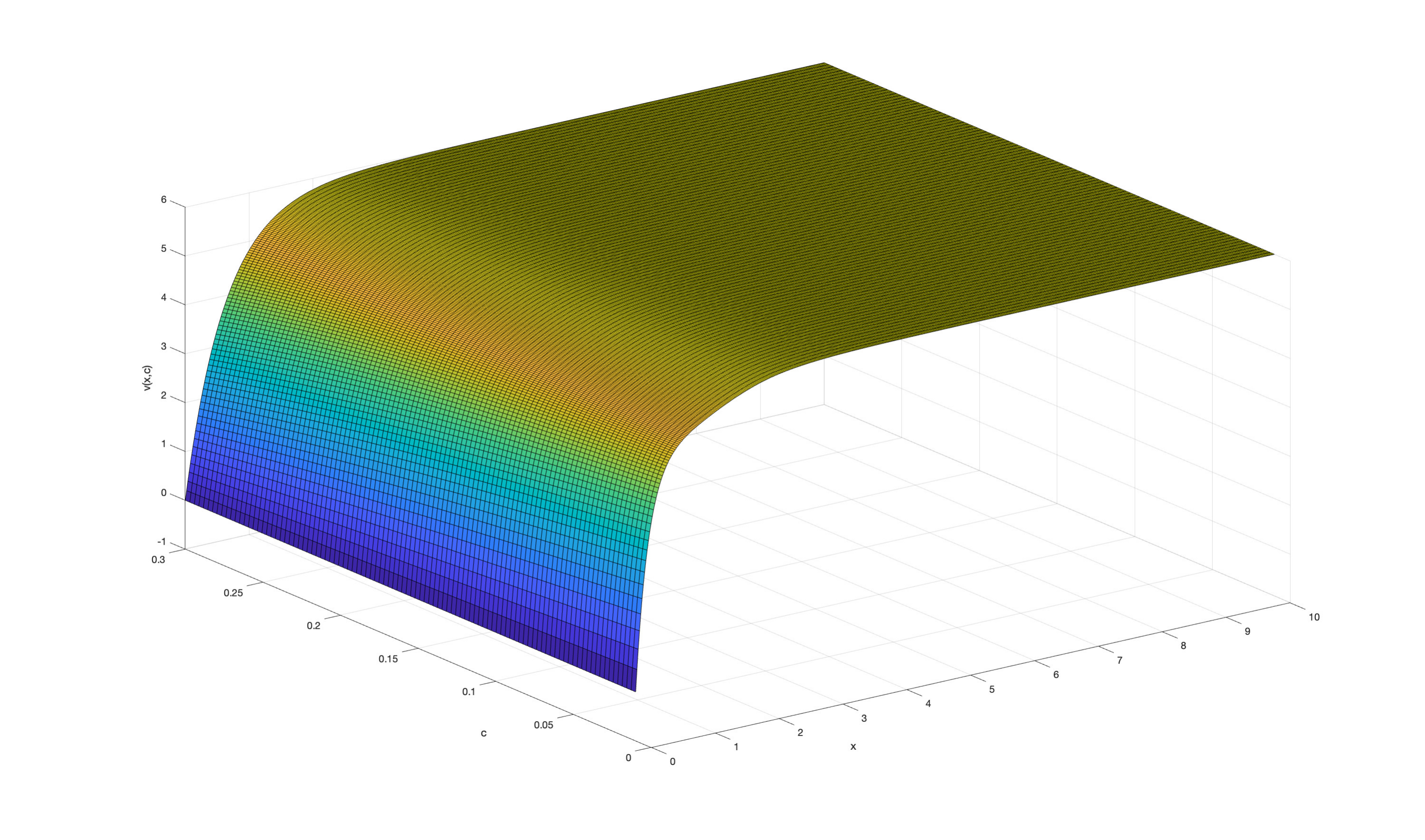}
\caption{The value function $v$ where $\mu=0.4$, $r=0.05$, $\sigma=0.4$, $\cc=0.3 $.}
\label{fig:b1v}
\end{figure}
The following result fully characterizes the free boundary, which is crucial to determine the optimal dividend ratcheting strategy.
\begin{theorem}
\label{thm:freeboundary}
Let $v$ be a solution to \eqref{v_pb} given in \thmref{thm:u}. Define the dividend ratcheting free boundary:
$$\sw(c)=\inf\big\{x>0 \;\big|\; v_c(x,c)=0\big\},~ c\in[0,\cc].$$
Then $\sw(c)>0$ for all $c\in[0,\cc].$ Also,
\begin{enumerate}
\item it holds for all $(x,c)\in\Q$ that $v_{c}(x,c)=0$ if $x\geq \sw(c)$ and $v_{c}(x,c)<0$ if $x< \sw(c)$; 
\item it holds that $\sw(\cdot)\in C^\infty[0,\cc]$; 
\item there exit two constants $0<K_{1}<K_{2}$ such that $K_{1}\leq \sw'(c)\leq K_{2}$ for all $c\in[0,\cc]$. As a consequence, both $\sw(\cdot)$ and its inverse $\sw^{-1}(\cdot)$ are strictly increasing and Lipschitz continuous.
\end{enumerate}
\end{theorem}
The claims 1, 2 and 3 in \thmref{thm:freeboundary} will be proved, respectively, in \lemref{lem:xc}, \propref{smooth} and \propref{lip}. These properties will be used to construct an optimal dividend ratcheting strategy to the problem \eqref{value} in the next section.

\begin{figure}[H]
\centering
\includegraphics[width=0.8\linewidth]{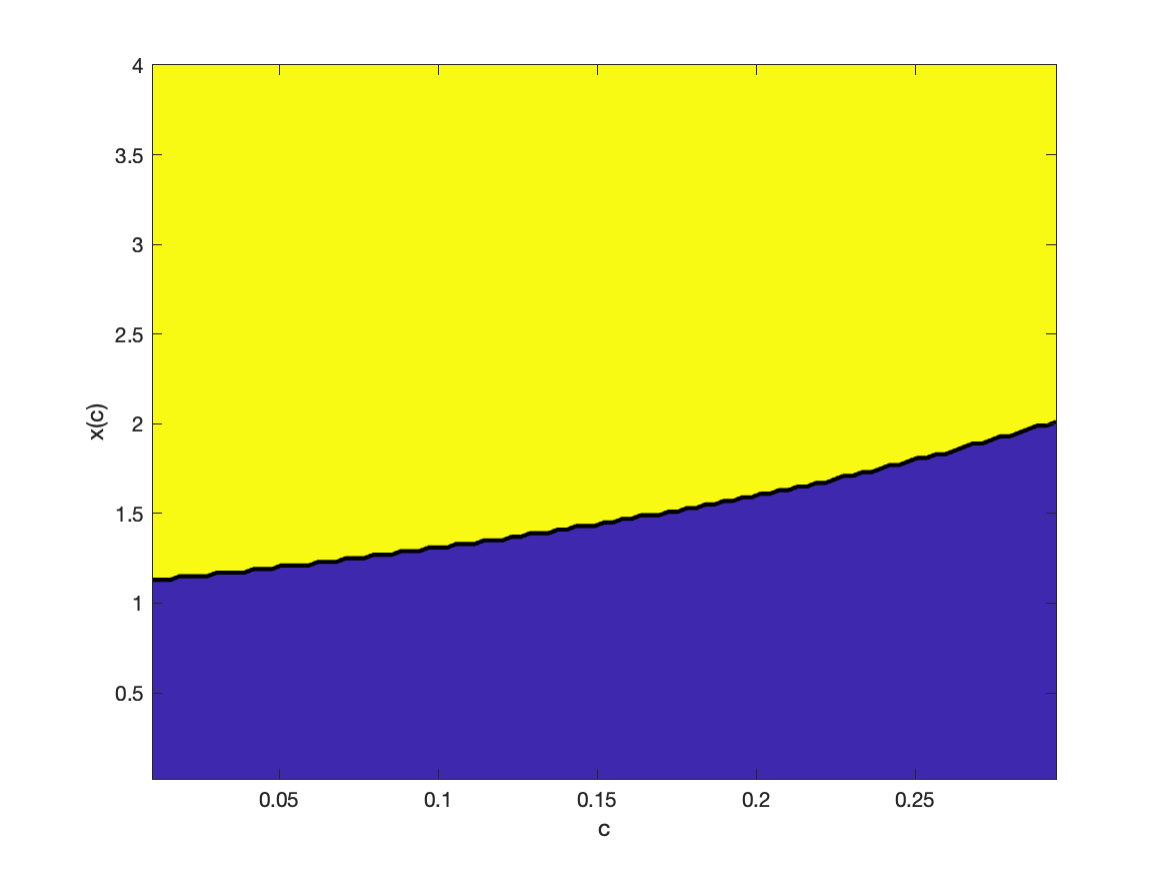}
\caption{The free boundary $\sw(\cdot)$ where $\mu=0.4$, $r=0.05$, $\sigma=0.4$, $\cc=0.3 $.}
\label{fig:fb1}
\end{figure}
Figure \ref{fig:fb1} illustrates the dividend ratcheting free boundary $\sw(\cdot)$.
Above the free boundary, we have $v_{c}=0$, which means it will not change the optimal value if one ratchet up the dividend payout rate up to the boundary.
In other words, the function $v$ takes a constant value on each horizon line on the left of the free boundary (where $v_{c}=0$). In particular, we have $v(x,c)=v(x,\cc)=g(x)$ if $x\geq \sw(\cc)$.
By contrast, the function $v$ is strictly decreasing on each horizon line on the right of the free boundary (where $v_{c}<0$), therefore, one should not ratchet up the dividend payout rate below the free boundary since it will reduce the optimal value.

\subsection{Optimal Dividend Ratcheting Strategy}\label{sec:verify}
Since $\sw(\cdot)$, given in \thmref{thm:freeboundary}, is strictly increasing, we may define
\begin{align}\label{swi}
\swi(x)=\begin{cases}
c,&0\leq x\leq \sw(c);\\
\sw^{-1}(x),&\sw(c)<x< \sw(\cc);\\
\cc,&x\geq \sw(\cc).
\end{cases}
\end{align}
By \thmref{thm:freeboundary}, the function $\swi(\cdot)$ is Lipschitz continuous, bounded and non-decreasing on $\R^{+}$.
Also, by the definition \eqref{swi},
\begin{align}\label{swi2}
\sw\big(\swi(x)\big)=\max\big\{\sw(c),~x\big\} \hbox{\; if $\swi(x)<\cc$.}
\end{align}

We are now ready to provide a complete answer to the problem \eqref{value}.
\begin{theorem}
\label{thm:averi}
Let $v$ be the solution to \eqref{v_pb} given in \thmref{thm:u}.
Then it is the optimal value of the optimal dividend ratcheting problem \eqref{value}.
Moreover,
$\{\DD^*_t\}_{t\geq 0}$ is an optimal dividend ratcheting strategy to the problem \eqref{value}, where
\[\DD^*_t :=\swi\Big(\max\limits_{r\in[0,t]}X^*_r\Big),\]
and $\swi(\cdot)$ is defined in \eqref{swi}, and
$X^*_t$ is the unique strong solution to the following SDE:
\begin{align}\label{optimalstate}
X^*_t=x+\int_0^t\Big(\mu-\swi\big(\max\limits_{r\in[0,s]}X^*_r\big)\Big) \ds+ \sigma W_t.
\end{align}
\end{theorem}

\begin{proof}
Suppose $V(x,c)$ is the value function defined in \eqref{value} and $v(x,c)$ is the solution to the HJB equation \eqref{v_pb} given in \thmref{thm:u}. We come to prove $v(x,c)=V(x,c)$ for any $(x,c)\in \Q$.

For any admissible dividend ratcheting strategy $\{\DD_t\}_{t\geq 0}\in \Pi_{[c,\cc]}$, let $X_{t}$ be the corresponding solution to \eqref{X_eq} with the initial value $X_{0}=x$. Let $T>0$ be any constant.
Then It\^o's formula gives
\begin{align*}
v(x,c)
&=\E [e^{-r(\tau\wedge T) }v(X_{\tau\wedge T},\DD_{\tau\wedge T})]
-\E \Bigg[\int_0 ^{\tau\wedge T }e^{-r t } \LL_{\DD_t}v(X_t,\DD_t) \dt\Bigg]\\
&\quad\;-\E \Bigg[\int_0 ^{\tau\wedge T }e^{-r t } v_c (X_t,\DD_t)\d \DD_t\Bigg]
-\E \Bigg[\int_0 ^{\tau\wedge T }e^{-r t }\si v_x (X_t,\DD_t) \d W_t\Bigg].
\end{align*}
Thanks to the boundedness of $v_{x}$, the last expectation is zero.
Since $-\LL_{\DD_t}v(X_t,\DD_t)-\DD_t\geq 0$, $-v_c (X_t,\DD_t)\geq 0$ and $\d \DD_t\geq 0$, we have
\begin{align}\label{vE}
v(x,c)
\geq
\E [e^{-r(\tau\wedge T)}v(X_{\tau\wedge T},\DD_{\tau\wedge T})]
+\E \Bigg[\int_0 ^{\tau\wedge T}e^{-r t } \DD_t \dt\Bigg].
\end{align}
Since $v\geq 0$, the first expectation can be dropped. 
Since $\DD_t$ is nonnegative, applying the monotone convergence theorem to the second integral leads to
\begin{align*}
v(x,c) \geq \E \Bigg[\int_0 ^{\tau}e^{-r t } \DD_t \dt\Bigg].
\end{align*}
Since $\{\DD_t\}_{t\geq 0}\in \Pi_{[c,\cc]}$ is arbitrary selected, we obtain $v(x,t)\geq V(x,t)$.

Because $\swi(\cdot)$ is Lipschitz continuous and bounded, by \cite[Theorem 2.2, p.150]{M08}, there is a unique strong solution $X^*_t$ to the SDE \eqref{optimalstate}.
Set $$\DD^*_t =\swi\Big(\max\limits_{r\in[0,t]}X^*_r\Big).$$
Then it is not hard to check that $\{\DD^{*}_t\}_{t\geq0}$ is an admissible dividend ratcheting strategy in $\Pi_{[c,\cc]}$.
Clearly, $X^*_{t}$ is the corresponding solution to \eqref{X_eq} under the strategy $\{\DD^*_t\}_{t\geq 0}$ with the initial value $X_{0}=x$.
We now prove that $\{\DD^*_t\}_{t\geq 0}$ is an optimal dividend ratcheting strategy to the problem \eqref{value} and $v(x,c)=V(x,c)$.

We first prove
\begin{align}\label{LV*}
-\LL_{\DD^*_t} v(X^*_t,\DD^*_t)-\DD^*_t=0.
\end{align}
In fact, if $\DD^*_t=\cc$, then \eqref{LV*} holds true since
$$-\LL_{\cc} v(x,\cc)-\cc=-\LL_{\cc} g(x)-\cc=0,~ x\in\R^{+}.$$
Now suppose $\DD^*_t<\cc$. Then by \eqref{swi2},
$$\sw(\DD^*_t)=\max\Big\{\sw(c),~\max\limits_{r\in[0,t]}X^*_r\Big\},
$$
so
$X^*_t\leq \sw(\DD^*_t).$ 
It thus follows from \thmref{thm:freeboundary} and \eqref{-Lv=0} that \eqref{LV*} holds.

On the other hand, we have $\d \DD^*_t =0$ if $X^*_t\neq \sw(\DD^*_t)$, and thanks to \thmref{thm:freeboundary}, $v_c(X^*_t,\DD^*_t)=0$ if $X^*_t=\sw(\DD^*_t)$,
so it always holds that
$$-v_c(X^*_t,\DD^*_t)\d \DD^*_t=0,~ 0\leq t\leq \tau.$$
Therefore, under the strategy $\{\DD^*_t\}_{t\geq 0}$, the inequality \eqref{vE} becomes an equation
\begin{align}\label{vE2}
v(x,c)=
\E [e^{-r(\tau\wedge T) }v(X^*_{\tau\wedge T},\DD^*_{\tau\wedge T})]
+\E \Bigg[\int_0 ^{\tau\wedge T }e^{-r t } \DD^*_t \dt\Bigg].
\end{align} 
If $\tau<\infty$, then $X^*_{\tau}=0$ and hence,
\begin{align*}
\lim_{T\to\infty}e^{-r(\tau\wedge T) } v(X^*_{\tau\wedge T},\DD^*_{\tau\wedge T})=
e^{-r\tau}v(X^*_{\tau},\DD^*_{\tau})=e^{-r\tau} v(0,\DD^*_{\tau})=0.
\end{align*}
If $\tau=\infty$, then, since $v$ is bounded and $r>0$,
\begin{align*}
\lim_{T\to\infty}e^{-r(\tau\wedge T) } v(X^*_{\tau\wedge T},\DD^*_{\tau\wedge T})=
\lim_{T\to\infty}e^{-r T } v(X^*_{T},\DD^*_{T})=0.
\end{align*} 
Now applying the dominated convergence theorem to the first integral and the monotone convergence theorem to the second integral in \eqref{vE2}, we get
$$
v(x,c)=\E \Bigg[\int_0 ^{\tau }e^{-r t } \DD^*_t \dt\Bigg].
$$
In view of the definition of $V(x,t)$, we obtain $v (x,t)\leq V(x,t)$.
Therefore, $v (x,t)=V(x,t)$, and $\{\DD^*_t\}_{t\geq 0}$ is an optimal dividend ratcheting strategy.
\end{proof}

Different from the simple case, we can see from this result that, in the complicated case, the optimal dividend ratcheting strategy relays on the surplus level and one should only ratchet up the dividend payout rate when the surplus level touches the dividend ratcheting free boundary.

We will prove Theorems \ref{thm:u} and \ref{thm:freeboundary} by a novel PDE method in Sections \ref{section5} and \ref{section6}. Before doing that, we first present a numerical analysis of the problem in Section \ref{section4}.

\section{Numerical Analysis}\label{section4}

We preform a numerical study in this section to exam the effects of different parameters (including the maximum allowable dividend payout rate $\cc$ in Figure \ref{comparecbar}, the discount rate $r$ in Figure \ref{comparer}, the income rate $\mu$ in Figure \ref{comparemu} and the volatility rate $\sigma$ in Figure \ref{comparevol}) on the value function $V$ and the free boundary $\sw(\cdot)$. 
\begin{figure}[H]
\centering
\includegraphics[width=0.4785\linewidth]{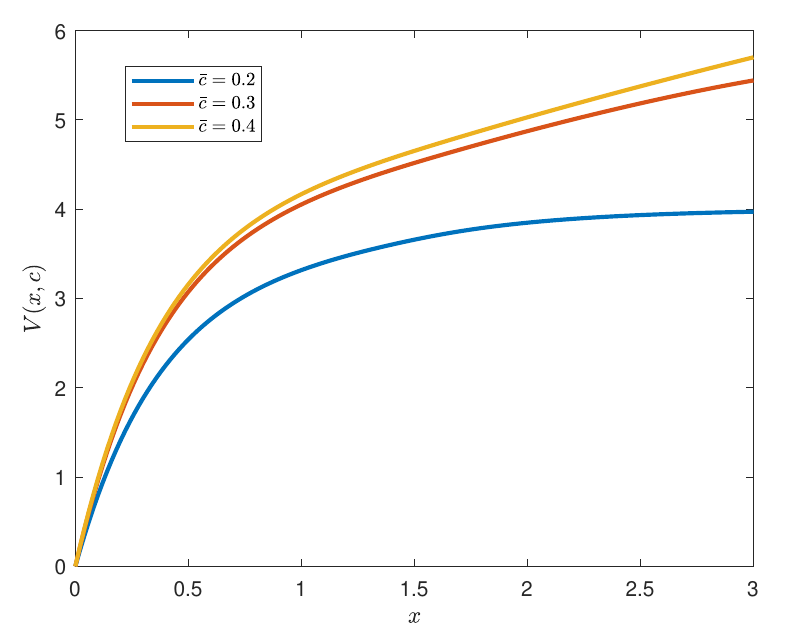} 
\includegraphics[width=0.485\linewidth]{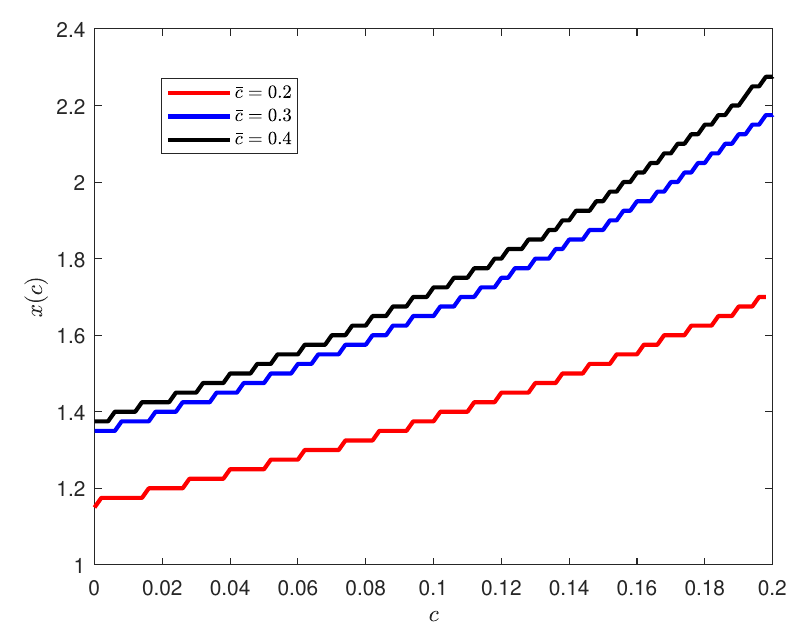}
\caption{The value function $V(x,c)$ and free boundary $\sw(\cdot)$ with different $\cc$ where $c=0.1$, $\mu=0.4$, $r=0.05$, $\sigma=0.4$.}
\label{comparecbar}
\end{figure}
Figure \ref{comparecbar} exams the effect of the maximum allowable dividend payout rate $\cc$ on the value function $V$ and the free boundary $\sw(\cdot)$. 
The upper panel displays the value functions $V$ (sliced at $c=0.1$). As expected, the bigger the maximum rate $\cc$, the bigger the value. 
If the current smallest rate $c=0.1$ is far from the maximum rate $\cc$ (say $\cc=0.3$ or $0.4$), 
the value function has a higher speed of increasing with respect to the surplus $x$, meaning that one can do significantly better if a higher surplus is given. This is because one has a large room to increase the dividend payout rate. By contrast, there is no much increasing with respect to the surplus in the value function when the smallest and maximum rates are not far. 
The lower panel depicts the free boundary $\sw(\cdot)$. As expected, they are increasing in all scenarios. But there are no much differences in the speed of increasing, in particular, when the current smallest rate $c=0.1$ is far from the maximum rate $\cc$.

\begin{figure}[H]
\centering
\includegraphics[width=0.47\linewidth]{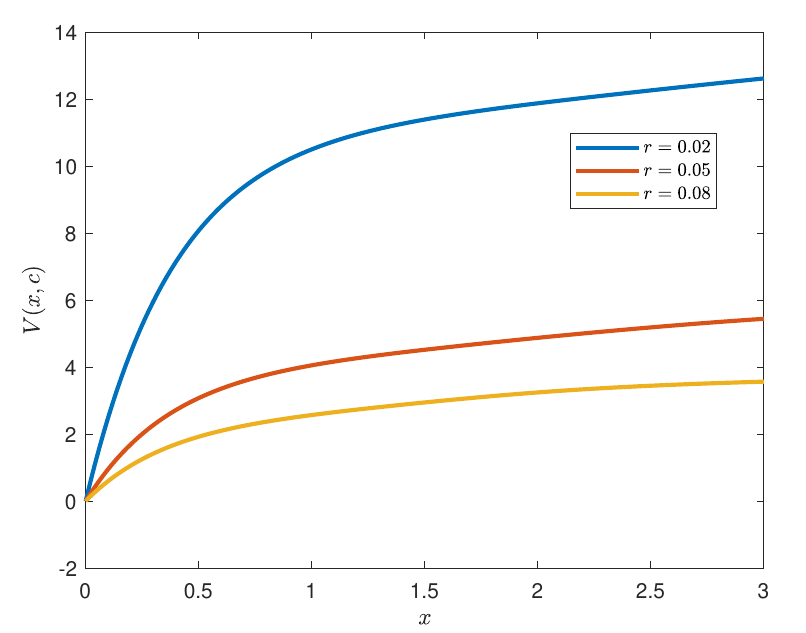} \includegraphics[width=0.455\linewidth]{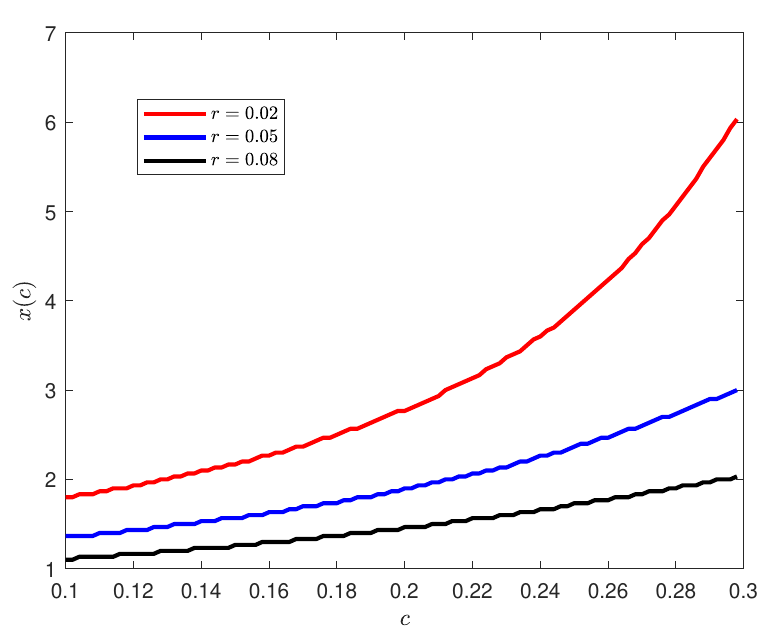}
\caption{The value function $V(x,c)$ and free boundary $\sw(\cdot)$ with different $r$ where $c=0.1$, $\cc=0.3$, $\mu=0.4$, $\sigma=0.4$.} 
\label{comparer}
\end{figure}

Figure \ref{comparer} exams the effect of the discount rate $r$ on the value function $V$ and the free boundary $\sw(\cdot)$. 
The upper panel displays the value functions $V$ (sliced at $c=0.1$). As expected, the smaller the rate $r$, the bigger the value. It can seen from the figure that the effect of $r$ becomes less important as the surplus $x$ becomes bigger. 
The lower panel depicts the free boundary $\sw(\cdot)$. From the figure we see that the smaller the discount rate $r$, the more significant the changes of the free boundary close to the maximum dividend payout rate $\cc$. Therefore, a smaller discount rate may have significant impact on the dividend payout strategy. Intuitively speaking, if a company hopes to apply a stable dividend payout strategy (that is, less sensitive with respect to the surplus level), the company should choose a bigger discount rate $r$.

\begin{figure}[H]
\centering
\includegraphics[width=0.475\linewidth]{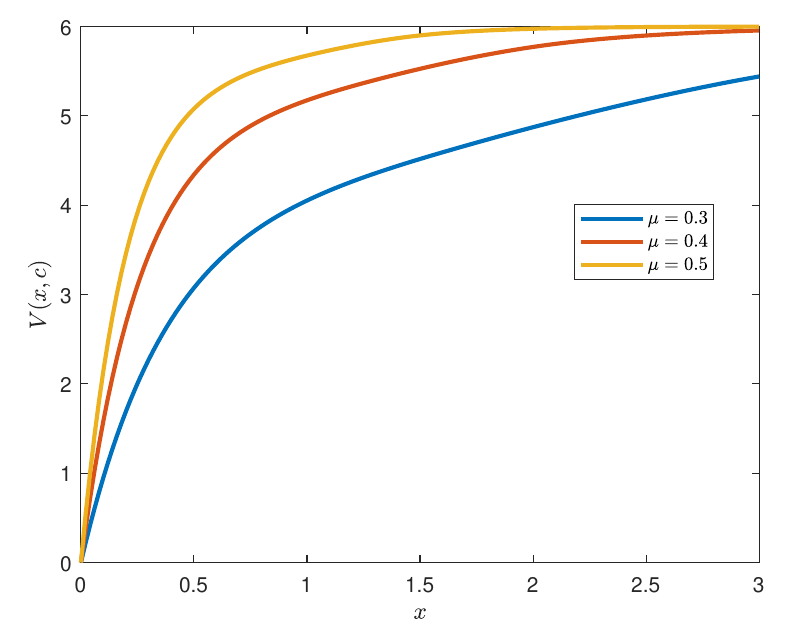} 
\includegraphics[width=0.4885\linewidth]{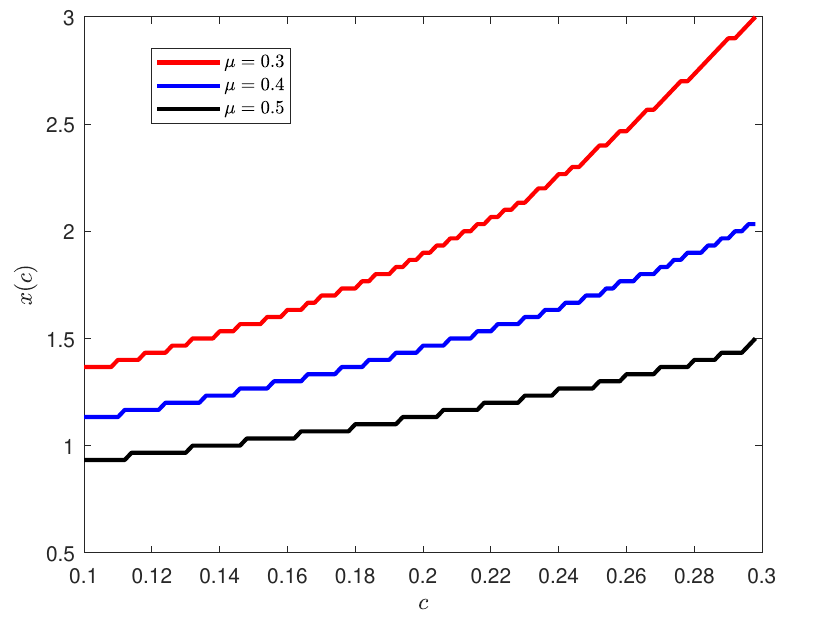}
\caption{The value function $V(x,c)$ and free boundary $\sw(\cdot)$ with different $\mu$ where $c=0.1$, $\cc=0.3$, $r=0.05$, $\sigma=0.4$.} 
\label{comparemu}
\end{figure}

Figure \ref{comparemu} exams the effect of the income rate $\mu$ on the value function $V$ and the free boundary $\sw(\cdot)$. The upper panel displays the value functions $V$ (sliced at $c=0.1$). As expected, the bigger the rate $\mu$, the bigger the value. It can seen from the figure that the effect of $\mu$ becomes less important as the surplus $x$ becomes bigger. 
The lower panel depicts the free boundary $\sw(\cdot)$. 
The figure shows that the dividend payout strategy is more stable if the company's income rate is higher. Intuitively speaking, one should invest into those companies with stable dividend payout strategies since their income rates should be higher.

\begin{figure}[H]
\centering
\includegraphics[width=0.475\linewidth]{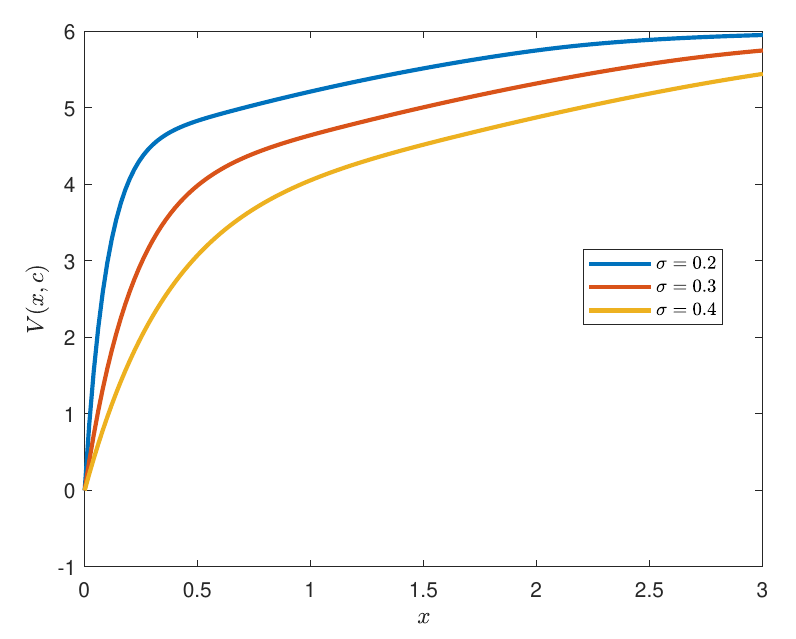} \includegraphics[width=0.485\linewidth]{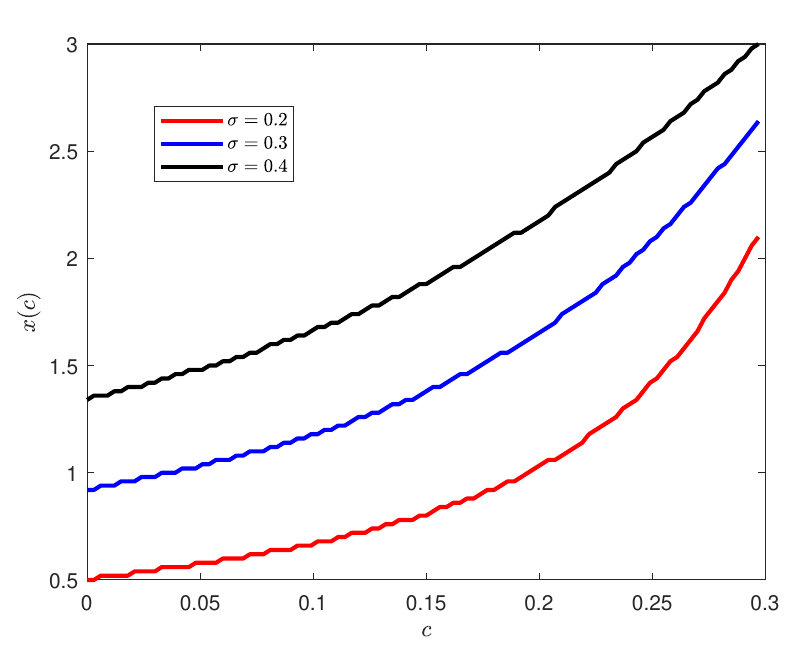}
\caption{The value function $V(x,c)$ and free boundary $\sw(\cdot)$ with different $\sigma$ where $c=0.1$, $\cc=0.3$, $r=0.05$, $\mu=0.3$.} 
\label{comparevol}
\end{figure}

Figure \ref{comparevol} exams the effect of the volatility $\sigma$ on the value function $V$ and the free boundary $\sw(\cdot)$. The upper panel displays the value functions $V$ (sliced at $c=0.1$). As expected, the smaller the volatility $\sigma$, the bigger the value. But the differences are not so significant. 
The lower panel depicts the free boundary $\sw(\cdot)$. The figure shows that the dividend payout strategy is more stable if the company's volatility rate is smaller. Intuitively speaking, one should invest into those companies with stable dividend payout strategies since their volatility rates should be smaller.

\section{Solvability of the HJB Equations \eqref{v_pb} and \eqref{u_pb}}\label{section5}
This section consists of two parts. In the first part, we introduce and study a regime switching system to approximate the PDEs \eqref{v_pb} and \eqref{u_pb}. In the second part, we construct a solution to \eqref{v_pb} and \eqref{u_pb} by a limit argument. We start with a technical lemma that will be frequently used in our subsequent analysis. 
\begin{lemma}\label{lem:tech}
Suppose $a\geq 0$ is a constant, and $h: [a,\infty) \to \R^{+}$ is a given bounded measurable function. If $\nu\in \Wploc[a,\infty)\cap L^\infty([a,\infty))$ for some $p>1$ and satisfies
\begin{align}\label{nu_pb}
\begin{cases}
\min\big\{\!-\LL_c \nu+h,~\nu\big\}=0, & x>a,\medskip\\
\nu(a)=0,
\end{cases}
\end{align}
for some $c\in[0,\cc]$, then $\nu\equiv 0$ in $[a,\infty)$.
\end{lemma}
\begin{proof}
It is easy to check that $\nu(x)\equiv 0$, $x\geq a$ satisfies \eqref{nu_pb}, by the uniqueness of the solution of \eqref{nu_pb} in $\nu\in \Wploc[a,\infty)\cap L^\infty([a,\infty))$, we get the conclusion.
\end{proof}

\subsection{A Regime Switching Approximation System}\label{sec:approximation}

Suppose $n\geq 1$ and the dividend payout rates can only take the following finite values $c_i=\cc- i\Dc,~ i=0,1,2,\cdots,n,$ with $\Dc=\cc/n>0$.
Let $v_{-1}=-1$. Consider a regime switching ODE system:
\begin{align}\label{vi_pb}
\begin{cases}
\min\big\{\!-\LL_{c_i} v_i- c_i,~v_i-v_{i-1}\big\}=0, & x>0, \medskip\\
v_i(0)=0, & i=0,1,2,\cdots,n,
\end{cases}
\end{align}
under the power growth condition.
Here, we can think of $v_i(x)$ as an approximation of $v(x,c_i)$.
This is a system of single-obstacle problems, so we can solve it.

\begin{lemma}\label{lem:vi}
The system \eqref{vi_pb} has a unique solution $v_i\in \Wploc(\R^+)\cap C^{1+\al}(\R^+)$, $i=0,1,2,\cdots,n,$ for any $p>1$ and $\al\in(0,1)$. Moreover, it holds that
\begin{gather}\label{vi}
0\leq v_i\leq \frac{\cc}{r},\\\label{vix}
v_i' \geq 0, \\
\label{vixx}
v_i'(y)\leq \max\big\{v_i'(x),~1\big\},~\forall\; 0\leq x\leq y.
\end{gather}
Let
\begin{align}
x_{i}=\inf\big\{x> 0\;|\;v_i(x)=v_{i-1}(x) \big\},~ i=0, 1,2,\cdots,n. \label{def:xi}
\end{align}
Then $v_i(x)=v_{i-1}(x)$ if $x\geq x_i$ and $v_i(x)>v_{i-1}(x)$ if $x<x_i$.
Also, $0<x_{i}<\infty$ for $i=1,2,\cdots,n,$ and $x_{0}=+\infty$.
\end{lemma}

\begin{proof}
Because $g\geq 0>v_{-1}$ and $c_0=\cc$, it follows from \eqref{g_eq} that
\begin{align*}
\begin{cases}
\min\big\{\!-\LL_{c_0} g- c_0,~g-v_{-1}\big\}=-\LL_{\cc} g- \cc=0, & x>0, \\
g(0)=0. &
\end{cases}
\end{align*}
Therefore, $v_0=g$ is the unique solution to \eqref{vi_pb} in $\Wploc(\R^+)\cap C^{1+\al}(\R^+)$ and $x_{0}=+\infty$.

For each $i=1,2,\ldots, n$, the problem \eqref{vi_pb} is a single-obstacle problem. 
By the standard penalty method and $L^p$ theorem, we can prove step by step for $i$ from $1$ to $n$ that
\eqref{vi_pb} admits a unique solution (See, e.g., \cite{Fr75}):
$v_i\in \Wploc(\R^+). $
Moreover, by the embedding theorem (See, e.g., \cite[Theorem 6 on page 286]{Ev17}), we also have
$v_i\in C^{1+\al}(\R^+).$
Since the process is standard, we omit the details.

By \eqref{vi_pb},
\begin{align}\label{increasingvi}
v_i\geq v_{i-1}\geq\cdots \geq v_0\geq 0>v_{-1},
\end{align}
giving the lower bound in \eqref{vi}.

We now prove the remaining claims by mathematical induction.
It is easy to check that $v_0=g $ satisfies all the desired properties.
Suppose all the desired results of the lemma hold for all $i\leq j-1$ (with $1\leq j<n$), we now prove that they also hold when $i=j$.

Because $v_{j-1}\leq \cc/r$, the constant function $\cc/r$ is a super solution to the variational inequality of $v_j$ by \eqref{vi_pb}. Hence we proved the upper bound in \eqref{vi}: $v_{j}\leq \cc/r$.

We now prove the set $$ E_{j}:=\{x> 0\;|\;v_j(x)=v_{j-1}(x) \}$$ is not empty. Suppose, on the contrary, it is empty. Then by the variational inequality \eqref{vi_pb}, we have
$$-\LL_{c_j} v_j- c_j=0,~ x>0.$$
It admits an explicit solution $ v_j(x)=\frac{ c_j }{r}(1-e^{-\ga_j x}), $
under the initial condition $v_j(0)=0$ and power growth condition, where $\ga_j$ is the positive root of
$
-\frac{1}{2}\si^2 \ga_j^2+(\mu-c_j) \ga_j+r=0.
$
It hence follows $v_j(+\infty)=c_j/r<\cc/r=v_0(+\infty),$ contradicting to the order \eqref{increasingvi}.
Therefore, $E_{j}$ is not empty. As a consequence, we have
$x_j:=\inf E_{j}<\infty.$

We next prove
\begin{align}\label{vix<=1}
v_{j-1}'(x)\leq 1,~ x\geq x_j.
\end{align}
Thanks to \eqref{vixx}, it suffices to prove
$v_{j-1}'( x^*)\leq 1$ for all $ x^*\in E_j$.
Suppose, on the contrary, $ v_{j-1}'( x^*) > 1$ for some $x^*\in E_j$.
Then by the order \eqref{increasingvi}, there must exist $1\leq k\leq j$ such that $v_{j}( x^* )=v_{j-1}( x^* )=\cdots=v_{j-k}( x^* )>v_{j-k-1}( x^* ).$
Since $v_{j-k}$ and $v_{j-k-1}$ are continuous, there is a neighborhood $I_\ep=[ x^*-\ep, x^*+\ep]$ such that $v_{j-k}>v_{j-k-1}$ in $I_\ep$, so by \eqref{vi_pb},
$-\LL_{c_{j-k}} v_{j-k}-c_{j-k}=0\quad\hbox{in } I_\ep.$
Together with the variational inequality of $v_j$, we obtain the estimate
\begin{align}\label{v''}
\frac{\si^2}{2}(v_j-v_{j-k})''\leq -(\mu-c_j) (v_j'-v_{j-k}')+r(v_j-v_{j-k})+k\Dc (1- v_{j-k}')
~\hbox{ a.e. in } I_\ep.
\end{align}
Note that both $v_j-v_{j-1}$ and $v_j-v_{j-k}$ attain their minimum value $0$ at $ x^*$, so
$$v_{j}( x^* )=v_{j-1}( x^* )=v_{j-k}( x^* ),~ v_j'( x^* )=v_{j-1}'( x^* )=v_{j-k}'( x^*)>1.$$
which shows that the right hand side (RHS) of \eqref{v''} is negative at $ x^*$. Because the RHS of \eqref{v''} is continuous, we conclude that $(v_j-v_{j-k})''<0$ a.e. in $I_\ep$ when $\ep>0$ is small enough. This means $v_j-v_{j-k}$ is strictly concave in $I_\ep$, which contradicts that $v_j-v_{j-k}$ attains its minimum value at the inner point $ x^* \in I_\ep$. Thus, \eqref{vix<=1} follows.

We next prove
\begin{align}\label{vjj}
v_j(x)=v_{j-1}(x),~ x\geq x_j.
\end{align} 
Let $\nu=v_j-v_{j-1}$. Then, by \eqref{vi_pb}, 
\begin{align*}
\begin{cases}
\min\big\{\!-\LL_{c_j} \nu-\LL_{c_j} v_{j-1}- c_j,~\nu\big\}=
\min\big\{\!-\LL_{c_j} v_{j}- c_j,~v_{j}-v_{j-1}\big\}=0, & x>x_j,\medskip\\
\nu(x_j)=v_{j}(x_j)-v_{j-1}(x_j)=0.
\end{cases}
\end{align*}
Note \eqref{vi_pb} also holds when $i=j-1$. Together with \eqref{vix<=1}, we get
\begin{align*}
-\LL_{c_j} v_{j-1}(x)- c_j
=-\LL_{c_{j-1}} v_{j-1}(x)- c_{j-1}+\Dc ( 1 - v_{j-1}' (x))\geq 0,~ x>x_j,
\end{align*}
Applying \lemref{lem:tech}, we conclude $\nu=0$ for $x>x_j$. 
As a consequence, \eqref{vjj} holds and $x_j$ is the unique free boundary point such that $v_j(x)=v_{j-1}(x)$ if $x\geq x_j$ and $v_j(x)>v_{j-1}(x)$ if $x<x_j$.

Note that $v_j\geq g$, $v_j(0)=g(0)=0$ and $g'(0)>1$ in the complicated case, so
\begin{align}\label{vj'0>1}
v_j'(0)\geq g'(0)>1,
\end{align}
which together with \eqref{vix<=1} gives
\begin{align}\label{xj>0}
x_j>0.
\end{align}

We come to prove $v_j'\geq 0$. Since $v_j=v_{j-1}$ in $[x_j,+\infty)$ by \eqref{vjj} and $ v_{j-1}'\geq 0$ by the inductive hypothesis, we only need to prove
\begin{align}\label{vjx>=0}
v_j'\geq 0\quad\hbox{in}\; [0, x_j).
\end{align}
Differentiating the equation in \eqref{vi_pb} we have
\begin{align}\label{vjx_eq}
-\LL_{c_j} ( v_j')=0\quad\hbox{in}\; (0,x_j).
\end{align}
Notice \eqref{vj'0>1} and $v_j'(x_j)=v_{j-1}'(x_j)\geq 0$ by \eqref{vjj},
we conclude from the maximum principle that \eqref{vjx>=0} holds.

Finally, fix any $x\in\R^{+}$. We now prove, for all $y\geq x$,
\begin{align}\label{vjxx}
v_j'(y)\leq \max\big\{v_j'(x),~1\big\}.
\end{align}
Combining \eqref{vix<=1} and \eqref{vjj}, we know \eqref{vjxx} is true for $y\in [x_j,+\infty)$.
Applying \eqref{vix<=1}, it is easy to check that the constant function $\max\big\{v_j'(x),~1\big\}$ is a super solution to \eqref{vjx_eq} in $[x,x_j]$, so \eqref{vjxx} holds for $y\in [x,x_j]$ as well.
The proof is complete.
\end{proof}
In the rest of this paper, we keep the notations $v_{i}$ and $x_{i}$ given in \lemref{lem:vi}.

We next prove that the inequality \eqref{vix<=1} holds strictly at $x=x_j$.
\begin{lemma}\label{lem:vj'<1}
For each $j=1,2,\cdots,n$, we have
\begin{align}\label{vj'<1}
v_{j-1}'(x_j)< 1.
\end{align}
\end{lemma}
\begin{proof}
Suppose, on the contrary, suppose $v_{j-1}'(x_j )=1$.
Same as before, there must exist $1\leq k\leq j$ such that
$v_{j}(x_j )=v_{j-1}(x_j )=\cdots=v_{j-k}(x_j )>v_{j-k-1}(x_j ).$
By continuity, $v_{j-k}>v_{j-k-1}$ near $x_{j}$, so
$-\LL_{c_{j-k}} v_{j-k}- c_{j-k}=0$ by \eqref{vi_pb}, which confirms the continuity of $v_{j-k}''$ near $x_j $.
Note that both $v_j-v_{j-1}$ and $v_j-v_{j-k}$ attain their minimum value $0$ at $x_j$, so
\begin{align}\label{vj'=1}
v_j'(x_j )=v_{j-1}'(x_j )=v_{j-k}'(x_j )=1.
\end{align}
Together with \eqref{vixx} we obtain $v_{j-k}''(x_j )\leq 0$.
Moreover, it follows from \eqref{vj'0>1} and \eqref{vixx} that $v_{j-k}'(x)\leq v_{j-k}'(0)$ for all $x\geq0$; and consequently, $v_{j-k}''(0)\leq 0$.
Differentiating the equation of $v_{j-k}$ in $[0,x_j ]$ twice we obtain
$-\LL_{c_{j-k}} (v_{j-k}'')=0.$
By the maximum principle, we deduce
$v_{j-k}''\leq 0$ in $[0,x_j ]$, so $v_{j-k}'(x)\geq v_{j-k}'(x_j )=1$ for $x\in[0,x_j ]$.

Let $\psi=v_j-v_{j-k}.$
By the equations of $v_j$ and $v_{j-k}$, we see that $\psi$ satisfies
$-\LL_{c_j} \psi=k\Dc (v_{j-k}'-1)\geq0 \quad\hbox{in } [0,x_j ].$
Since $\psi$ attains its minimum value 0 at $x_j $, we conclude from the Hopf lemma that $ \psi' (x_j -)<0$, but this contradicts \eqref{vj'=1}.
The proof is thus complete.
\end{proof}

The next result establishes the monotonicity of the free boundaries $x_i$.
To this end, we define
$$u_i:=\frac{v_i-v_{i-1}}{\Dc },~ i=1,2,\cdots,n.$$

\begin{lemma}\label{lem:xi}
For each $i=1,2,\cdots,n,$ we have
\begin{align}\label{ui_pb}
\begin{cases}
\min\big\{\!-\LL_{c_i} u_i- v_{i-1}'+1,~u_i\big\}=0, & x>0.\medskip\\
u_i(0)=0,
\end{cases}
\end{align}
Moreover, 
\begin{align}\label{xi}
x_i\leq x_{i-1},~ i=1,2,3,\cdots,n.
\end{align}
\end{lemma}
\begin{proof}
By \lemref{lem:vi}, $x_i$ is the number such that $u_{i}(x)>0$ when $x<x_{i}$ and $u_{i}(x)=0$ when $x\geq x_{i}$.

By taking a difference between the equations of $v_1$ and $v_0$, we know \eqref{ui_pb} holds when $i=1$. Clearly, \eqref{xi} holds when $i=1$.
Now suppose \eqref{ui_pb} holds when $i=j-1$ ($2\leq j<n$). We are going to prove $u_j$ satisfies \eqref{ui_pb} and \eqref{xi} holds when $i=j$.

Suppose $\wu$ is the unique bounded solution to the variational inequality
\begin{align}\label{wu_pb}
\begin{cases}
\min\big\{\!-\LL_{c_j} \wu- v_{j-1}'+1,~\wu\big\}=0, & x>0,\medskip\\
\wu(0)=0,
\end{cases}
\end{align}
and let $\wx=\inf\big\{x>0\;|\;\wu(x)=0\big\}.$
It suffices to prove $u_j=\wu$, $x_j=\wx$ and $\wx\leq x_{j-1}$.

We first prove that
\begin{align}\label{vj-1'<=1}
v_{j-1}'(x)\leq 1 \quad \hbox{for all } x\geq \wx.
\end{align}
By virtue of \eqref{vixx}, it suffices to prove $v_{j-1}' (x^*)\leq 1$ for any $x^*>0$ such that $\wu(x^*)=0$.
Suppose, on the contrary, $v_{j-1}'(x^*) > 1$. By \eqref{wu_pb} we have
\begin{align}\label{uj''}
\frac{\si^2}{2}\wu''\leq -(\mu-c_j) \wu'+r \wu- v_{j-1}'+1.
\end{align}
Note that $\wu$ attains its minimum value $0$ at $x^*$, so $\wu(x^*)=\wu'(x^* )=0.$
Then the RHS of \eqref{uj''} is continuous and negative at $x^*$, so we have $\wu''<0$ in a neighborhood of $x^*$. This means that $\wu$ is strictly concave in that neighborhood, which contradicts that $\wu$ attains its minimum value at the inner point $x^*$ of the neighborhood. Hence, we proved \eqref{vj-1'<=1}.

We next prove that $\wu>0$ for $x<\wx$ and $\wu=0$ for $x\geq \wx$.
By the definition of $\wx$, it only needs to prove
\begin{align}\label{wu=0}
\wu(x)=0,~ x\in[\wx,+\infty).
\end{align}
Indeed by \eqref{wu_pb}, we have
$$
\begin{cases}
\min\big\{\!-\LL_{c_j} \wu- v_{j-1}'+1,~\wu\big\}=0, & x>\wx,\medskip\\
\wu(\wx)=0.
\end{cases}
$$
We thus conclude from \eqref{vj-1'<=1} and \lemref{lem:tech} that \eqref{wu=0} holds.

Next, we come to prove
\begin{align}\label{wx}
\wx\leq x_{j-1}.
\end{align}
To this end, 
let $\ou$ be the unique solution in $W^{2, p}([0,x_{j-1}])$ of the following ODE
\begin{align}\label{ou_pb}
\begin{cases}
-\LL_{c_j} \ou - v_{j-1}'+1=0, & 0<x<x_{j-1},\medskip\\
\ou (x_{j-1})=\ou' (x_{j-1})=0,
\end{cases}
\end{align}
and let $\ou=0$ for $x>x_{j-1}$. Then $\ou\in W^{2, p}(\R^+)$.
If we can prove that $\ou$ is a super solution to \eqref{wu_pb}, then \eqref{wx} follows.

Thanks to \eqref{vj'<1} and \eqref{vjj}, we can check that
$$-\LL_{c_j} \ou- v_{j-1}'+1=- v_{j-1}'+1=- v_{j-2}'+1\geq 0\quad \hbox{in}\; [x_{j-1}. \infty).$$
Together with \eqref{ou_pb}, we get
\begin{align}\label{Lou}
-\LL_{c_j} \ou- v_{j-1}'+1\geq 0,~ x>0.
\end{align}
So in order to prove $\ou$ is a super solution to \eqref{wu_pb}, it only needs to prove
\begin{align}\label{ou}
\ou(x)\geq 0,~ x\in[0, x_{j-1}].
\end{align}
The proof of \eqref{ou} is divided into the following four steps.

\textbf{Step 1:}
Note that \eqref{vj'<1} and the equation of $u_{j-1}$ imply
$\frac{\si^2}{2} u_{j-1}'' (x_{j-1}-)=- v_{j-2}' (x_{j-1})+1>0.$
So there is a small $\ep>0$ such that
\begin{align}\label{uj''>0}
u_{j-1}'' (x)>0,~ x\in (x_{j-1}-\ep, x_{j-1}).
\end{align}

\textbf{Step 2:}
Let
$w=\frac{\ou-u_{j-1}}{\Dc },~ x\in[0,x_{j-1}].$ 
We claim there is a small $\ep>0$ such that
\begin{align}\label{wix<0}
w'(x)< 0,~ x\in(x_{j-1}-\ep, x_{j-1}).
\end{align}
Indeed, from the equation of $u_{j-1}$ and $\ou$ we know
\begin{align}\label{wi_pb}
\begin{cases}
-\LL_{c_j} w=2 u_{j-1}', & 0<x<x_{j-1},\medskip\\
w(x_{j-1})=w'(x_{j-1})=0.
\end{cases}
\end{align}
Differentiating the equation in \eqref{wi_pb} and applying \eqref{uj''>0} we have
$$-\LL_{c_j} w'=2 u_{j-1}'' >0,~ x\in (x_{j-1}-\ep, x_{j-1}).$$
Since $w'(x_{j-1})=0$, if $w'(x)\geq 0$ for some $x\in(x_{j-1}-\ep, x_{j-1})$, then by the maximum principle and the Hopf Lemma, we have $w''(x_{j-1}-)<0$. But by the equation in \eqref{wi_pb},
$$\frac{\si^2}{2} w''(x_{j-1}-)=\Big(-2 u_{j-1}' -(\mu-c_j) w'+r w\Big)(x_{j-1})=0,$$
leading to a contradiction. So we established \eqref{wix<0}.

\textbf{Step 3:}
We get from \eqref{wix<0} and $w(x_{j-1})=0$ that, 
for some small $\ep>0$,
\begin{align}\label{w>0}
w(x)> 0,~ x\in(x_{j-1}-\ep, x_{j-1}).
\end{align}

\textbf{Step 4:}
We claim
\begin{align}\label{w>=0}
w(x)\geq 0,~ x\in[0, x_{j-1}].
\end{align}
This together with $u_{j-1}\geq 0$ will lead to the desired estimate \eqref{ou}.

Suppose, on the contrary, \eqref{w>=0} is not true.
Let
$$x^*=\sup\big\{x\in [0,x_{j-1}]\;\big|\; w(x)\leq 0\big\}.$$ 
Then $0<x^*<x_{j-1}$ by virtue of \eqref{w>0}.
By the continuity of $w$, we have $w(x^*)=0$. Integrating the equation in \eqref{wi_pb} in $[x^*, x_{j-1}]$, we have
$$\frac{\si^2}{2} w'(x^*)=-r \int_{x^*}^{x_{j-1}}w (x)\d x -2u_{j-1}(x^*) <0.$$
It follows $w(x^*+\ep)<w(x^*)=0$ for sufficiently small $\ep>0$, which contradicts to the definition of $x^*$.
Therefore, \eqref{w>=0} is true and the desired estimate \eqref{ou} is established.

It is left to prove $\wu=u_j$, which would imply that $\wx=x_j$ because $\wx$ and $x_j$ are the minimum roots for $\wu$ and $u_j$, respectively.

Let $\wv=v_{j-1}+\wu\Dc$. We come to prove that $\wv $ satisfies the same variational inequality as $v_j$, namely
\begin{align}\label{vi+1_pb}
\left\{
\begin{array}{ll}
\min\big\{\!-\LL_{c_j} \wv - c_j,~\wv -v_{j-1}\big\}=0, & x>0, \medskip\\
\wv (0)=0.
\end{array}
\right.
\end{align}
By the uniqueness of this variational inequality, we then have $\wv=v_j$ which is equivalent to the desired equation $\wu=u_j$.

Firstly, $\wv -v_{j-1}=\wu\Dc\geq0$ is clear. Secondly, since
\begin{gather}\label{vi+1_ineq}
-\LL_{c_{j-1}} v_{j-1} - c_{j-1} \geq 0,\\
\label{u_ineq}
\big(-\LL_{c_j} \wu- v_{j-1}'+1\big)\Dc \geq0,
\end{gather}
adding them up yields
$-\LL_{c_j} \wv - c_j \geq0.$
So
$$\min\big\{\!-\LL_{c_j} \wv - c_j,~\wv -v_{j-1}\big\}\geq0.$$
Finally, if $\wv (x)>v_{j-1}(x)$, namely $\wu (x)>0$ at some $x>0$, then $x\in (0,\wx)$. By \eqref{wu_pb}, we see \eqref{u_ineq} is an equation at $x$.
By \eqref{wx}, we have $x<\wx\leq x_{j-1}$ so that \eqref{vi+1_ineq} is also an equation at $x$. Consequently, $-\LL_{c_j} \wv- c_j=0$ at $x$ so that $\wv $ satisfies \eqref{vi+1_pb}.
This completes the proof.
\end{proof}


\begin{lemma}\label{lem:vix_ub}
There is a constant $K > 0$ that is independent of $n$ and $i$ such that
\begin{align}\label{vix_ub}
0\leq v_i' \leq K,~ i=0,1,\cdots,n.
\end{align}
\end{lemma}
\begin{proof}
The case $i=0$ is evident. We now consider the case $i=1,\cdots,n$.
The lower bound in \eqref{vix_ub} has already established in \eqref{vixx}.
To establish the upper bound, we notice that the following variational inequality (see \cite{Ta00})
\begin{align*} 
\begin{cases}
\min\big\{\!-\LL_{0} \ov,~\ov_x-1\big\}=0, & x>0,\medskip\\
\ov(0)=0,
\end{cases}
\end{align*}
admits a unique solution
\[\ov(x)=
\begin{cases}
K_1 ({\rm e}^{\theta_2 x}-{\rm e}^{\theta_1 x}),& 0\leq x<x_\infty;\\
K_2+x, &x_\infty\leq x<+\infty,
\end{cases}
\]
where $\theta_1<0<\theta_2$ are the roots of
$-\frac{1}{2}\si^2\theta^2-\mu\theta+r=0,$ 
and
\begin{align*}
x_\infty &=\frac{2}{\theta_2-\theta_1}\ln\Big|\frac{\theta_1}{\theta_2}\Big|>0,\\
K_1&=(\theta_2e^{\theta_2x_\infty}-\theta_1e^{\theta_1x_\infty})^{-1}>0,\\
K_2&=K_1 ({\rm e}^{\theta_1x_\infty}-{\rm e}^{\theta_2x_\infty})-x_\infty<0.
\end{align*}
Clearly, $$-\LL_{c_i} \ov-c_i=-\LL_{0} \ov+c_i(\ov_x-1)\geq 0,$$
so $\ov$ is a super solution to \eqref{vi_pb}.
Since $v_i(0)=\ov(0)=0$, it follows that $v_i' (0)\leq \ov' (0)$.
Then we uniformly have $v_i'(x) \leq \max\big\{ \ov' (0),~1\big\}$ by \eqref{vixx}, completing the proof by taking, say, $K=\max\{ \ov' (0),1\}$.
\end{proof}

\begin{lemma}\label{lem:ui_b}
For each $p>1$, there is a constant $K_{p}>0$, which is independent of $N\geq 1$, $n $ and $i$, such that
\begin{align}\label{ui_b0}
0\leq u_i&\leq K_{p},\\
\label{ui_b}
|u_i|_{W^{2, p}([N-1,N])}&\leq K_{p},~ i=0,1,\cdots,n.
\end{align}
\end{lemma}
\begin{proof}
We can rewrite the problem \eqref{ui_pb} as
\begin{align*} 
\begin{cases}
-\LL_{c_i} u_i=( v_{i-1}' -1)\mathbf{1}_{\{x<x_i\}}, & x>0,\medskip\\
u_i(0)=0,
\end{cases}
\end{align*}
By \eqref{vix_ub}, we know that $|( v_{i-1}' -1)\mathbf{1}_{\{x<x_i\}}|$ is uniformly bounded. Applying the maximum principle, we obtain \eqref{ui_b0}. Then the $L^p$ estimation gives \eqref{ui_b} (See e.g., \cite[Theorem 9.13 on page 239]{TG11}). 
\end{proof}

By the Sobolev embedding theorem we also have
\begin{corollary}\label{cor:uix_b}
For each $0<\al<1$, we have $ u_i' \in C^\al(\R^+)$ and
\begin{align}\label{ui_b1}
| u_i' |_{C^\al(\R^+)}\leq K_{\al},~ i=0,1,\cdots,n.
\end{align}
where $K_{\al} > 0$ is a constant independent of $n$ and $i$.
\end{corollary}

\begin{lemma}\label{lem:ui_b2}
We have
\begin{align}\label{ui_b2}
\Big|\frac{u_i-u_{i-1}}{\Dc }\Big|\leq K,~ i=1,\cdots,n,
\end{align}
where $K > 0$ is a constant independent of $n$ and $i$.
\end{lemma}
\begin{proof}
By \corref{cor:uix_b},
$K :=\frac{1}{r} \max\limits_{i} \sup\limits_{x\in \R^+}| u_i' |<\infty.$
Let $\ou_i=u_{i-1}+K \Dc $, then
\begin{align*}
-\LL_{c_i} \ou_i- v_{i-1}'+1
=&-\LL_{c_i} u_{i-1}+r K \Dc - v_{i-1}'+1\\
=&-\LL_{c_{i-1}} u_{i-1}- v_{i-2}'+1- u_{i-1}' \Dc+r K \Dc \\
\geq&- u_{i-1}' \Dc+r K \Dc \geq 0.
\end{align*}
Therefore, $\ou_i$ is a super solution to \eqref{ui_pb}, giving $u_i\leq \ou_i=u_{i-1}+K \Dc $.
Similarly, we can prove $u_i\geq u_{i-1}-K \Dc $. Therefore, \eqref{ui_b2} follows.
\end{proof}

\subsection{Proof of \thmref{thm:u}}\label{sec:solvability}

Now we are ready to prove \thmref{thm:u}. \medskip\\ 
For each $n\in Z^+$, rewrite $u_i(x)$ and $v_i(x)$ as $u^n_i(x)$ and $v^n_i(x)$, respectively.
Let $u^n(x,c)$ and $v^n(x,c)$ be the linear interpolation functions of $u^n_i(x)$ and $v^n_i(x)$, respectively.
Note \lemref{lem:ui_b}, \corref{cor:uix_b} and \lemref{lem:ui_b2} imply $\{u^n\}$ are uniformly bounded and Lipschitz continuous in $\Q$. Applying the Arzela-Ascoli theorem, there is $u\in C(\Q )$, and a subsequence $\{u^{n_k}\}\subset\{u^n\}$ such that, for each $L>0$, the sequence $\{u^{n_k}\}$ converges to $u$ in $C([0,L]\times [0,\cc])$.
Also, it is easy to prove that 
$u$ and $v$ satisfy the relationship \eqref{v_u}. Clearly, $u$ is nonnegative and bounded by \eqref{ui_b0}.
Moreover, \lemref{lem:ui_b} implies $u(\cdot,c), v(\cdot,c)\in \Wploc(\R^+)$ for any $c\in [0,\cc]$ and $p>1$. Also, \eqref{vi} and \eqref{ui_b0} imply $v$ and $u$ are bounded in $\Q $.

We come to prove that $(u,v)$ defined above satisfies, for every $c\in [0,\cc]$,
\begin{align}\label{u_pb2}
\min\big\{\!-\LL_c u-v_{x}+1,~u\big\}=0, \hbox{\; a.e.\;in $\R^+$.}
\end{align}
For each $c\in [0,\cc]$, there is a sequence $c^k=\cc-i_k \cc/n_k$ such that $c^k\to c$ and
$u^{n_k}_{i_k}(\cdot)\to u(\cdot,c)$ in $C[0,L]$
for any $L>0$.
Moreover, from \lemref{lem:ui_b} and \corref{cor:uix_b} we have
$u^{n_k}_{i_k}(\cdot)$ or its subsequence $\to u(\cdot,c)~\hbox{weakly in }W^2_{p }[0,L]\hbox{ and uniformly in }C^{1+\al}[0,L]$
for any $L>0$.
Let $k\to \infty$ in the inequality
$$-\LL_{c^k} u^{n_k}_{i_k}-\p_x v^{n_k}_{i_k-1}+1\geq 0$$
we get
$$-\LL_c u(\cdot,c)-v_x(\cdot,c)+1\geq 0,\hbox{\; a.e.\;in $\R^+$.}$$
Hence,
$$\min\big\{\!-\LL_c u-v_x+1,~u\big\}\geq 0,\hbox{\; a.e.\;in $\R^+$.}$$
On the other hand, suppose $u>0$ at some point $(x,c)\in\Q$. Fix this $c$. Then by continuity $u^{n_k}_{i_k}(x)> 0$ for $k$ large enough. Thanks to \lemref{lem:vi}, $u^{n_k}_{i_k}(y)> 0$ for all $y\leq x$.
It hence follows
$$-\LL_{c^k} u^{n_k}_{i_k}(y)-\p_x v^{n_k}_{i_k-1}(y)+1=0\hbox{\; for all $y\leq x$.}$$
Taking limit yields
$$-\LL_c u(y,c)-v_x(y,c)+1=0\hbox{\; for a.e. $y\leq x$.}$$
Hence, we proved that \eqref{u_pb2} holds.
The estimates \eqref{v}-\eqref{vxx} follow from \eqref{vi}-\eqref{vixx}, \eqref{vix_ub}, \eqref{ui_b0}, estimate \eqref{ub} from \eqref{ui_b0} and \eqref{ui_b1}, and estimate \eqref{ucb} from \eqref{ui_b2}.

In particular, the continuity of $u_{x}$, $v_{x}$ and $v_{c}$ follows from \eqref{ub} and \eqref{v_u}. Since $u$ and $v$ satisfy \eqref{v_u} and $u>0$ is an open set, \eqref{u_pb2} implies that $u$ and $v$ are smooth in the region $u>0$. It is then easy to verify that $u$ satisfies \eqref{u_pb} in the sense of Definition \ref{hjb2}
and $v$ satisfies \eqref{v_pb} in the sense of Definition \ref{hjb1}.
The proof of \thmref{thm:u} is thus complete.


\section{Properties of the dividend ratcheting free boundary $\sw(\cdot)$}\label{section6}
Let $(u,v)$ be given in \thmref{thm:u} throughout this section.
Recall that $\sw(\cdot)$ is given in \thmref{thm:freeboundary}:
\begin{align*}
\sw(c)=\inf\big\{x>0\mid v_{c}(x,c)=0\big\}=\inf\big\{x>0\mid u(x,c)=0\big\},~ c\in[0,\cc].
\end{align*}
We now establish the following \propref{smooth} and \propref{lip}.
\begin{proposition}\label{smooth}
It holds that $\sw(\cdot)\in C^\infty[0,\cc].$
\end{proposition}
\begin{proposition}\label{lip}
There exit two constants $0<K_{1}<K_{2}$ such that $K_{1}\leq \sw'(c)\leq K_{2}$ for all $c\in[0,\cc]$.
\end{proposition}

\subsection{On the smoothness of the dividend ratcheting free boundary}\label{sec:smoothfreeboundary}

We prove \propref{smooth} in this section.
More precisely, we will use a bootstrap method to establish the smoothness of $\sw(\cdot)$.
To this end, we first prove several lemmas.

\begin{lemma} 
If $\sw(c)<\infty$ for some $c\in[0,\cc]$, then
\begin{align}\label{vxc<=1}
v_x(x,c)\leq 1\quad\hbox{for all } x\geq \sw(c).
\end{align}
\end{lemma}
\begin{proof}
Assume $\sw(c)<\infty$ for some $c\in[0,\cc]$.
To establish \eqref{vxc<=1}, it suffices to prove $v_x(\sw(c),c)\leq 1$ by \eqref{vxx}.
By the definition of $\sw(c)$, there is a positive sequence $\{x_k\}$ such that $u(x_k,c)=0$ and $\lim_{k}x_k=\sw(c)$.
Since $u(\cdot,c)$ attains its minimum value 0 at (the inner point) $x_k>0$, we have $u_x(x_k,c)=0$.
By \eqref{u_pb}, we have
\begin{align}\label{vxx<=}
\frac{\si^2}{2}u_{xx}(x,c)\leq -(\mu-c) u_x(x,c)+r u(x,c)+1-v_x(x,c).
\end{align}
Suppose $v_x(x_k,c)>1$. 
Then the RHS of \eqref{vxx<=} is continuous and negative at $x=x_k$, so we conclude $u_{xx}(\cdot,c)<0$ in a neighborhood of $x_k$. Hence, $u(\cdot,c)$ is strictly concave in the neighborhood. But this contradicts $u(\cdot,c)$ takes its minimum value 0 at the inner point $x_k$ of the neighborhood, so we have $v_x(x_k,c)\leq 1.$
Now, taking limit leads to the desired inequality $v_x(\sw(c),c)\leq 1$.
\end{proof}

\begin{lemma}\label{lem:xc}
For any $(x,c)\in\Q$, we have
$u(x,c)=0$ if $x\geq \sw(c)$ and $u(x,c)>0$ if $x< \sw(c)$.
Also,\begin{align}\label{xc>0}
\sw(c)>0,~ c\in[0,\cc].
\end{align}
\end{lemma}
\begin{proof}
By the definition of $\sw(\cdot)$, we have $u(x,c)>0$ when $x< \sw(c)$.
If $\sw(c)=\infty$, then $x< \sw(c)$ always holds. If $\sw(c)<\infty$, then
by virtue of \eqref{u_pb2}, \eqref{vxc<=1} and $u(\sw(c))=0$, we conclude from \lemref{lem:tech} that $u(x,c)=0$ if $x\geq \sw(c)$.

Notice that $v\geq g$, $v(0,c)=g(0)=0$ and $g'(0)>1$ in the complicated case, so
\begin{align}\label{vx0>1}
v_x(0,c)\geq g'(0)>1.
\end{align}
Because $\sw(c)\geq 0$ and \eqref{vxc<=1}, we get \eqref{xc>0}. 
The last assertion follows from \eqref{v_u}.
\end{proof}

\begin{lemma}\label{lem:x'c}
The curve $\sw(\cdot)$ is non-decreasing in $[0,\cc]$.
\end{lemma}
\begin{proof}
We argue by contradiction. Suppose there exist $0\leq a<b\leq \cc$ such that $\sw(b)<\sw(a)$. Then for any $x\in (0,\sw(a))$, by \lemref{lem:xc} and \eqref{v_pb}, we have
$-\LL_a v(x,a) - a=0.$
By \eqref{v_pb} and \eqref{u_pb},
$$-\LL_c v - c\geq 0,~ \p_c(-\LL_c v - c)=\LL_c u+v_x-1\leq 0, \hbox{ in $\Q$},$$
it hence follows
$$-\LL_c v(x,c) - c=0,~ -\LL_c u(x,c)-v_x(x,c)+1=0,~(x,c)\in (0,\sw(a))\times[a,\cc].$$
This particularly yields
$$0=-\LL_{b} v(x,b)-b,~ -\LL_b u(x,b)-v_x(x,b)+1=0,~ x\in (\sw(b),\sw(a)).$$
But we have $u(x,b)=0$ for all $x\in [\sw(b),\sw(a)]$, so the above leads to
$$0=-\LL_b u(x,b)-v_x(x,b)+1=-v_x(x,b)+1,~ x\in (\sw(b),\sw(a)),$$ 
and consequently,
$$0=-\LL_{b} v(x,b)-b=r v(x,b)-\mu,~ x\in (\sw(b),\sw(a)). $$
But the above two equations clearly cannot hold simultaneously,
completing the proof.
\end{proof} 

\begin{lemma}\label{lem:xcc}
The curve $\sw(\cdot)$ is bounded and continuous in $[0,\cc]$.
\end{lemma}
\begin{proof} 
To prove $\sw(\cdot)$ is bounded in $[0,\cc]$, it suffices to prove $\sw(\cc)$ is finite
since $\sw(\cdot)$ is positive and non-decreasing. 
Suppose $\sw(\cc)=+\infty$. Then 
$$u(0,\cc)=0,~ -\LL_{\cc} u(x,\cc)=v_x(x,\cc)-1=g'(x)-1=\frac{\cc \ga}{r}e^{-\ga x}-1,~ x>0.$$
This ODE with Dirichlet boundary condition
admits a unique (power growth) solution
$$u(x,\cc)=\frac{1}{r}(e^{-\ga x}-1)+\frac{\ga \cc}{r(\si^2 \ga-\mu+\cc)} x e^{-\ga x}.$$
It thus follows $u(+\infty,\cc)=-1/r<0$, contradicting to $u\geq 0$. Hence, we proved that $\sw(\cdot)$ is bounded in $[0,\cc]$.
Now we prove the continuity.
Suppose there is one $a\in [0,\cc]$ such that $\sw(a-)<\sw(a+)$. Then $u(x,a)=0$, $-\LL_a v(x,a) - a=0$ and $-\LL_a u(x,a)-v_x(x,a)+1=0$ for $x\in (\sw(a-),\sw(a+))$. A similar argument to the proof of \lemref{lem:x'c} leads to a contradiction. So $\sw(\cdot)$ is continuous.
\end{proof}

To establish the smoothness of $\sw(\cdot)$, 
define a compact region
$$\Om=\Big\{(x,c)\in \Q\;\Big|\; 0\leq x\leq \sw(c)\Big\},$$
and a function space
$$\seta=\bigg\{\varphi:\Om \to \R \;\bigg|\; \frac{\p^m \varphi}{\p x^m}\in C(\Om) \mbox{\; for all integers $m\geq 1$}\bigg\}.$$
For $\varphi\in \seta$, define the maximum norm
\[|\varphi(\cdot,c)|_{C^m[0,\sw(c)]}=\max_{0\leq x\leq \sw(c)}\bigg\vert\frac{\p^m \varphi}{\p x^m}(x,c)\bigg\vert.\]
Define another function space
\begin{align*}
\setb=\Big\{\varphi:\Om\to \R \;\Big|\; & \varphi\in C(\Om), 
\mbox{ and $\varphi(\cdot,c)\in \Wp([0,\sw(c)]) $}\\
&\mbox{for each $c\in[0,\cc]$ and each $p>1$}\Big\}.
\end{align*}

\begin{lemma}\label{lem:Cinf}
Suppose $f\in \seta$ and $h\in C[0,\cc]$.
Then there is a unique $\varphi\in\setb$ such that
\begin{align}\label{varphi_pb}
\begin{cases}
-\LL_c \varphi(x,c)=f(x,c),~& (x,c)\in \Om,\medskip\\
\varphi(0,c)=0,~
\varphi(\sw(c),c)=h(c),~& c\in[0,\cc].
\end{cases}
\end{align}
Moreover, $\varphi\in \seta$.
\end{lemma}
\begin{proof}
We omit the proof of the existence and uniqueness as it is standard. 

By a standard argument of the regularity of equation, we have for each $m\geq 1$,
\begin{align}\label{cm}
\sup_{c\in[0,\cc]}|\varphi(\cdot,c)|_{ C^m[0,\sw(c)]}=\sup_{(x,c)\in\Om}\bigg\vert\frac{\p^m \varphi}{\p x^m}(x,c)\bigg\vert<\infty.
\end{align}
Fix any $c\in[0,\cc)$ and $m\geq 1$. Suppose $ \Delta\in [0, \cc-c]$ and set
$$\psi(x)=\varphi(x,c+\Delta)-\varphi(x,c).$$
By the monotonicity of $\sw(\cdot)$ and \eqref{cm}, $\varphi_x(x,c+\Delta)$ is uniformly bounded for all $x\in[0,\sw(c)]$ and $ \Delta\in [0, \cc-c]$. This together with $f\in \seta$
and
$$-\LL_{c} \psi(x)=f(x,c+\Delta)-f(x,c)-\varphi_x(x,c+\Delta)\Delta,~ x\in[0, \sw(c)],$$
implies
$$\lim\limits_{\Delta\to 0+}|-\LL_{c} \psi|_{ C^m[0,\sw(c)]}=0.$$
Also,
$$\psi(\sw(c))=\varphi(\sw(c),c+\Delta)-\varphi(\sw(c),c)=h(c+\Delta)-h(c)-
\int_{\sw(c)}^{\sw(c+\Delta)}\varphi_x(y,c+\Delta)\dy,$$
so it follows from the continuity of $h(\cdot)$ and $\sw(\cdot)$ and boundedness of $\varphi_x$ that
$$\lim\limits_{\Delta\to 0+}\psi(\sw(c))=0.$$
Applying the regularity theorem of the equation, we obtain (please refer to \cite[Section 6.3.2 on Part II]{Ev17}))
$$\lim\limits_{\Delta\to 0+}|\psi|_{ C^m[0,\sw(c)]}=0,$$
so $\varphi\in \seta$.
\end{proof}

\begin{lemma}\label{lem:vuinf}
It holds that $v$, $u\in \seta$.
\end{lemma}
\begin{proof}
Since $-\LL_c v=c$ in $\Om$ and $v(\sw(\cdot),~\cdot)$ is continuous in $[0,\cc]$, we conclude from \lemref{lem:Cinf} that $v\in \seta$.
This further implies $-\LL_c u=v_x-1\in \seta$, which together with $u(\sw(c),c)\equiv 0$ and \lemref{lem:Cinf} confirms $u\in \seta$.
\end{proof}

We next further strengthen the estimate \eqref{vxc<=1} to the following.
\begin{lemma}\label{lem:vx<1}
there is a constant $0<\delta<1$, such that for all $c\in [0,\cc],$
\begin{align}\label{vx<1}
v_x(\sw(c),c)\leq 1-\delta.
\end{align}
\end{lemma}
\begin{proof}
Since $c\mapsto v_x(\sw(c),c)$ is a continuous function on $[0,\cc]$, it suffices to prove
that $v_x(\sw(c),c)<1$ for every $c\in[0,\cc]$.
Suppose, on the contrary, $v_x(\sw(c),c)=1$ for some $c\in[0,\cc]$.
We get from \lemref{lem:xc} and \eqref{v_pb} that
\begin{align}\label{Lv<=0}
\frac{\si^2}{2}v_{xx}(x-,c)=-c(1- v_x(x,c))-\mu v_x(x,c)+r v(x,c),~ x\leq \sw(c).
\end{align}
This together with \eqref{v} gives 
$$\frac{\si^2}{2}v_{xx}(\sw(c)-,c)=-\mu+r v(\sw(c),c)\leq -\mu+\cc\leq 0.$$
Also, by \eqref{vx0>1} and \eqref{vxx} we have $v_x(x,c)\leq v_x(0,c)$ for all $x\geq 0$, so $v_{xx}(0+,c)\leq 0$.
Differentiating \eqref{Lv<=0} w.r.t. $x$ twice gives
$$-\LL_c v_{xx}(\cdot,c)=0\quad\hbox{in }[0,\sw(c)].$$
Then applying the maximum principle yields
$v_{xx}(x,c)\leq 0$ for $x< \sw(c).$
This together with \eqref{u_pb} leads to
$$-\LL_c u(x,c)=v_x(x,c)-1 \geq v_x(\sw(c),c)-1=0,~ x\in [0,\sw(c)].$$
Since $u\geq0$ and $u(\sw(c),c)= 0$, the strong maximum principle gives $u_x(\sw(c),c)<0$. But $u(\cdot,c)$ gets its minimum value 0 at $\sw(c)$, so $u_x(\sw(c),c)=0$, contradicting the above. 
\end{proof}

\begin{lemma}\label{lem:uxx>0}
there is a constant $\delta>0$, such that for all $c\in [0,\cc],$
\begin{align}\label{uxx>0}
\frac{\si^2}{2} u_{xx}(\sw(c)-,c)\geq\delta.
\end{align}
Furthermore,
there is $0<\ep<\sw(0)$ such that 
\begin{align}\label{uxx>02}
\frac{\si^2}{2}u_{xx}(x-,c)\geq \frac{1}{2}\delta,~ u_{x}(y,c)-u_{x}(x,c)\geq \frac{\delta}{\si^2} (y-x)
\end{align}
for all $\sw(c)-\ep\leq x\leq y\leq \sw(c)$ and $c\in[0,\cc]$.
\end{lemma}
\begin{proof}
For $(x,c)\in\Om$, we have
$$\frac{\si^2}{2}u_{xx}(x-,c)=-(\mu-c) u_x(x,c)+r u(x,c)+1-v_x(x,c).$$
By \lemref{lem:vuinf}, the RHS is continuous and bounded in $\Om$, so is the left hand side.
Since $u_x(\sw(c),c)=u(\sw(c),c)=0$, the claim \eqref{uxx>0} follows from the above equation and \lemref{lem:vx<1}.
Recall that $\sw$ is continuous, increasing and positive on $[0,\cc]$, so it follows
$\frac{\si^2}{2}u_{xx}(x-,c)\geq \frac{1}{2}\delta$ 
in the compact set
$\{(x,c)\mid \sw(c)-\ep\leq x\leq \sw(c),~c\in[0,\cc]\},$ if $\ep\in(0, \sw(0))$ is sufficient small.
The second estimate in \eqref{uxx>02} then follows from the first one and the mean value theorem.
\end{proof}

\begin{lemma}\label{lem:x'c1}
The function $\sw(\cdot)$ is Lipschitz continuous in $[0,\cc]$, i.e. there is a constant $K>0$ such that for any $0\leq a< b\leq \cc$,
\begin{align}\label{Dxc}
|\sw(a)-\sw(b)|\leq K |a-b|.
\end{align}
\end{lemma}
\begin{proof}
Since $\sw(\cdot)$ is continuous,
it suffices to consider the case $0<|\sw(a)-\sw(b)|<\ep$, where $\ep$ is given in \lemref{lem:uxx>0}.
Denote $$\phi (x)=[u(x,b)-u(x,a)]/(b-a).$$ We first prove there is a constant $K>0$, which is independent of $a$ and $b$, such that
\begin{align}\label{wuxc}
|\phi_x (x)| \leq K,~ x\in [0,\sw(a)].
\end{align}
Note that $\phi$ satisfies
\begin{align*}
\begin{cases}
-\LL_{a} \phi=[v_x(x,b)-v_x(x,a)]/(b-a)-u_x(x,b), & x\in [0,\sw(a)],\medskip\\
\phi(0)=0,~ \phi(\sw(a))=[u(\sw(a),b)-u(\sw(a),a)]/(b-a),\\
\end{cases}
\end{align*}
First write
\begin{gather*}
[v_x(x,b)-v_x(x,a)]/(b-a)=\frac{1}{b-a}\int_{a}^{b}u_{x}(x,s)\ds,\\ 
\phi(x)=[u(x,b)-u(x,a)]/(b-a)=\frac{1}{b-a}\int_{a}^{b}u_{c}(x,s)\ds,
\end{gather*}
then by \eqref{ub}, \eqref{ucb}, we see that $|-\LL_{a} \phi|$ and $|\phi(\sw(a))|$ are uniformly bounded, independent of $a$ and $b$. For each $p>1$, by the maximum principle and the $L^p$ estimation, there is a constant $K_{p}>0$, which is independent of $a$ and $b$, such that
$|\phi|_{W^{2, p}([0,\sw(a)])}\leq K_{p}.$
Apply the embedding theorem, then \eqref{wuxc} follows. Consequently,
$$|u_x(\sw(a),b)-u_x(\sw(a),a)|\leq K |b-a|.$$
Meanwhile, by recalling $|\sw(b)-\sw(a)|<\ep$, it follows from \eqref{uxx>02} that
$$|u_x(\sw(b),b)-u_x(\sw(a),b)|\geq \frac{1}{2}\delta |\sw(b)-\sw(a)|.$$
Since $u_x(\sw(\cdot),\cdot)\equiv 0,$ we have
$$ u_x(\sw(b),b)-u_x(\sw(a),b)=u_x(\sw(a),a)-u_x(\sw(a),b).$$
Combining the above three estimates, we obtain \eqref{Dxc}.
\end{proof}

\begin{lemma}\label{lem:wuu}
Suppose $\varphi \in \setb$, $\varphi_x$ is bounded a.e. in $\Om$,
and
\begin{align}\label{cp}
\begin{cases}
\displaystyle{-\frac{\p (\LL_c \varphi)}{\p c}- \varphi_x=0,} & \hbox{ a.e. in $\Om$},\medskip\\
\displaystyle{\frac{\d [ \varphi(\sw(c),c)]}{\dc}- \varphi_x(\sw(c),c)\sw'(c)=0,}& \hbox{ a.e. $c\in[0,\cc]$},\medskip\\
\varphi(0,c)=0, & c\in[0,\cc],\medskip\\
\varphi(x,\cc)=0, & x\in[0,\sw(\cc)].
\end{cases}
\end{align}
Then $\varphi=0$ in $\Om$.
\end{lemma}
\begin{proof}
Denote $$\calE=\Big\{c\in [0,\cc]\;\Big|\; \varphi(x,c)=0 \hbox{\;for all\;} x\in [0,\sw(c)]\Big\}.$$
By the last equality in \eqref{cp}, $\cc\in \calE$.

Our target is to show $\calE=[0,\cc]$.
Now suppose $\calE\neq[0,\cc]$.
Then $a=\sup\big\{c\in[0,\cc]: c\notin\calE\big\}$ exists.
If $a=\cc$, then $a\in\calE$. Otherwise $a<\cc$, then $a+\ep\in\calE$ for all sufficiently small $\ep>0$. By the continuity of $\varphi$, we also have $a\in\calE$. Hence, we proved $[a,\cc]\subseteq\calE$.
Since $\calE\neq[0,\cc]$, we have $a>0$.

Let $I_\ep=[a-\ep, a]$, where $\ep\in(0,a)$ is a small constant to be chosen shortly. Also, let
$$\eta=\esssup\limits_{x\in [0,\sw(c)],~c\in I_\ep}|\varphi_x(x,c)|.$$
Since $ \varphi_x$ is bounded a.e. in $\Om$, we have $\eta<\infty$.
Thanks to the first equality in \eqref{cp}, we have
$$\sup\limits_{x\in [0,\sw(c)],~c\in I_\ep}\bigg|\frac{\p [\LL_c \varphi ]}{\p c}\bigg|\leq \eta.$$
Since $a\in\calE$, we have $\varphi(x,a)=0$ for all $x\in[0,\sw(a)]$. By integrating the above over $[c,a]$ for $c\in I_\ep$, it follows
$$\sup\limits_{x\in [0,\sw(c)],~c\in I_\ep}|\LL_c \varphi|\leq \eta \ep.$$
Moreover, from \eqref{Dxc} we obtain $|\sw'(\cdot)|\leq K$ a.e. for some positive constant $K>0$.
Since $\varphi(\sw(a),a)=0$, applying the second equality in \eqref{cp} leads to
\begin{align*}
\sup\limits_{c\in I_\ep}|\varphi(\sw(c),c)|=\sup\limits_{c\in I_\ep}\Bigg|-\int_{c}^{a} \varphi_x(\sw(s),s)\sw'(s) \ds\Bigg|
\leq K\eta \ep.
\end{align*}
Now using the maximum principle and $L^p$ estimation, there is a constant $K_1>0$, which is independent of $\eta$ and $\ep$, such that
$$\sup\limits_{c\in I_\ep}|\varphi(\cdot,c)|_{W^{2, p}[0,\sw(c)]}\leq K_1 \eta \ep.$$
Furthermore, the embedding theorem implies, for each $\al\in(0,1)$, there is a constant $K_{2}>0$, which is independent of $\eta$ and $\ep$, such that
$$\sup\limits_{c\in I_\ep}|\varphi(\cdot,c)|_{C^{1+\al}[0,\sw(c)]}\leq K_{2} \eta \ep.$$
Now choose $0<\ep<\min\{a, 1/K_{2}\}$. On recalling the definition of $\eta$ and that $\eta<\infty$, the above inequality leads to $\eta=0$.
This implies $I_\ep\subseteq \calE$, which contradicts the definition of $a$. Therefore, we must have $\calE=[0,\cc]$, completing the proof.
\end{proof}

\begin{lemma}\label{lem:w_pb}
there is a unique $w\in\setb$ such that
\begin{align}\label{w_pb}
\begin{cases}
-\LL_c w=-2 u_x, & \hbox{in \;} \Om,\medskip\\
w(0,c)=0,~
w(\sw(c),c)=0, &c\in [0,\cc].
\end{cases}
\end{align}
Moreover, $w\in \seta$ and $w=u_c$ in $\Om$.
\end{lemma}
\begin{proof}
Thanks to \lemref{lem:vuinf}, we have $-2u_x\in \seta$. Applying \lemref{lem:Cinf}, we see
$w\in \seta$.
To prove $w=u_c$ in $\Om$, it suffices to prove $\wu=u$ as $w=\wu_c$, where
$$\wu(x,c)=u(x,\cc)-\int_c^{\cc} w(x,s)\ds.$$
Indeed, we first have
$$
\frac{\p (-\LL_c \wu)}{\p c}- \wu_x=-\LL_c \wu_c=-\LL_c w=-2 u_x.
$$
Also, since $-\LL_c u=v_x-1,$
$$
\frac{\p (-\LL_c u)}{\dc}- u_x=\frac{\p (v_x-1)}{\p c}- u_x=-2 u_x.
$$
It hence follows $$\frac{\p (-\LL_c \wu)}{\p c}- \wu_x=\frac{\p (-\LL_c u)}{\p c}- u_x.$$
Next, we have $\wu_c(\sw(c),c)\equiv w(\sw(c),c)\equiv 0$. 
Since $u(\sw(c),c)\equiv u_x(\sw(c),c)\equiv0$, we also have $u_{c}(\sw(c),c)\equiv0$. 
We thus have
\begin{align*}
\frac{\d \wu(\sw(c),c)}{\dc}-\wu_x(\sw(c),c)\sw'(c)=\frac{\d u(\sw(c),c)}{\dc}-u_x(\sw(c),c)\sw'(c)~ \hbox{ a.e. } c\in[0,\cc].
\end{align*}
Clearly,
$\wu_{x}(x,c)=u_{x}(x,\cc)-\int_c^{\cc} w_{x}(x,s)\ds$ 
is bounded.
Therefore,
it follows from \lemref{lem:wuu} that $\wu=u$, completing the proof.
\end{proof}

\begin{lemma}\label{xc1}
It holds that $\sw(\cdot)\in C^1[0,\cc]$ and
\begin{align}\label{C1}
\sw'(c)=-\frac{u_{xc}(\sw(c),c)}{u_{xx}(\sw(c),c)},~ c\in[0,\cc].
\end{align}
\end{lemma}
\begin{proof}
By \lemref{lem:w_pb}, both $u_{xc}(\sw(c),c)$ and $u_{xx}(\sw(c),c)$ are continuous in $[0,\cc]$.
Notice $u_x(\sw(c),c)\equiv0,$ so we have
$$
0\equiv\frac{\d [u_x(\sw(c),c)]}{\dc}\equiv u_{xx}(\sw(c),c)\sw'(c)+u_{xc}(\sw(c),c).
$$
By virtue of \eqref{uxx>0}, the claim follows.
\end{proof}

Define $w^{(0)}(x,c)=u_{c}(x,c)$, 
$\sw^{(0)}(c)=\sw(c)$, and for $n=1,2,\cdots,$ 
$$w^{(n)}(x,c)=\frac{\p w^{(n-1)}(x,c)}{\p c},~\sw^{(n)}(c)=\frac{\d \sw^{(n-1)}(c)}{\dc}.$$

\begin{lemma}\label{lem:gn}
If $\sw(\cdot)\in C^n[0,\cc]$ and $w^{(n-1)} \in \seta$ for some $n\geq 1$, then we have $w^{(n)}(\sw(c),c)\in C[0,\cc]$.
\end{lemma}
\begin{proof}
Notice $w^{(0)}(\sw(c),c)\equiv u_{c}(\sw(c),c)\equiv 0$, so by taking its $n$-th derivative, we see that $w^{(n)}(\sw(c),c)$ can be expressed as a polynomial of
$$
\frac{\p w^{(k)}(\sw(c),c)}{\p x^m},~ k=0,1,\cdots,n-1,~ m=1,2,\cdots,n-k,
$$
and
$\sw^{(m)}(c),~ m=1,2,\cdots, n,$
which all are continuous in $[0,\cc]$, so $w^{(n)}(\sw(c),c)\in C[0,\cc]$.
\end{proof}

\begin{lemma}\label{lem:wn_pb}
If $\sw(\cdot)\in C^n[0,\cc]$ and $w^{(n-1)} \in \seta$ for some $n\geq 1$, then $w^{(n)}\in \seta$.
\end{lemma}
\begin{proof}
By \lemref{lem:gn} we know $w^{(n)}(\sw(c),c)\in C[0,\cc]$.
By \lemref{lem:Cinf}, we see there exits a unique $\phi\in\seta$ such that
\begin{align*} 
\begin{cases}
-\LL_c \phi=-(n+2) w^{(n-1)}_x & \hbox{in \;} \Om,\medskip\\
\phi(0,c)=0,~
\phi(\sw(c),c)=w^{(n)}(\sw(c),c) &c\in[0,\cc].
\end{cases}
\end{align*}
We come to prove $\phi=w^{(n)}$ in $\Om$, which will complete the proof of the lemma.

Let $$\Phi(x,c)=w^{(n-1)}(x,\cc)-\int_c^{\cc} \phi(x,s)\ds,$$
it suffices to prove $\Phi=w^{(n-1)}$ in $\Om$ since $\phi=\Phi_c$.

First, we have
$$
\frac{\p (-\LL_c \Phi)}{\p c}- \Phi_x=-\LL_c \Phi_c=-\LL_c \phi=-(n+2) w^{(n-1)}_x;
$$
and, since $-\LL_c w^{(n-1)}=-(n+1)w^{(n-2)}$,
$$
\frac{\p (-\LL_c w^{(n-1)})}{\p c}- w^{(n-1)}_x=-\frac{\p[(n+1)w^{(n-2)}]}{\p c}- w^{(n-1)}_x=-(n+2) w^{(n-1)}_x.
$$
Hence,
$$\frac{\p (-\LL_c \Phi)}{\p c}- \Phi_x=\frac{\p (-\LL_c w^{(n-1)})}{\p c}- w^{(n-1)}_x.$$
Second, we have $$\Phi_c(\sw(c),c)=\phi(\sw(c),c)=w^{(n)}(\sw(c),c)=w^{(n-1)}_c(\sw(c),c).$$ 
Thus,
\begin{align*}
\frac{\d \Phi(\sw(c),c)}{\dc}-\Phi_x(\sw(c),c)\sw'(c)=\frac{\d w^{(n-1)}(\sw(c),c)}{\dc}-w^{(n-1)}_x(\sw(c),c)\sw'(c).
\end{align*}
Finally, applying \lemref{lem:wuu} yields $\Phi=w^{(n-1)}$, completing the proof.
\end{proof}

Applying \eqref{C1}, we further have

\begin{lemma}\label{lem:xn+1}
If $\sw(\cdot)\in C^n[0,\cc]$ and $w^{(n)} \in \seta$ for some $n\geq 1$, then $\sw(\cdot)\in C^{n+1}[0,\cc]$.
\end{lemma}

Using $\sw^{(0)}(c)\in C^1[0,\cc]$ (by \lemref{xc1}), $w^{(0)}\in \seta$ (by \lemref{lem:w_pb}), \lemref{lem:wn_pb} and \lemref{lem:xn+1}, one can establish \propref{smooth} easily by mathematical induction.

\subsection{On the Lipschitz continuity of the free boundary}\label{sec:lipschitz}
Finally, we prove \propref{lip}.

\begin{lemma}\label{lem:xms}
We have $u_{xc}(\sw(c),c)\neq 0$ for all $c\in[0,\cc]$.
\end{lemma}
\begin{proof}
Write $w$ instead of $w^{(0)}$ for notation simplicity.
We then need to prove $w_x(\sw(c),c)\neq 0$.
Suppose, on the contrary, $w_x(\sw(c),c)=0$ for some $c\in [0,\cc]$.
Now fix this $c$.
It follows from \eqref{uxx>02} that
$$u_{xx}(x,c )>0,~ x\in (\sw(c)-\ep, \sw(c)).$$
Notice that $w=u_{c}$ satisfies \eqref{w_pb}, so differentiating the equation gives
$$-\LL_{c} w_x(x,c )=-2u_{xx}(x,c ) <0,~ x\in (\sw(c)-\ep, \sw(c)).$$
We claim
\begin{align}\label{wx>0}
w_x(x,c )> 0,~ x\in(\sw(c)-\ep, \sw(c)).
\end{align}
Indeed, if $w_x(x,c )\leq 0$ for some $x\in(\sw(c)-\ep, \sw(c))$, then since $w_x(\sw(c),c )=0$, the maximum principle and the Hopf Lemma imply $w_{xx}(\sw(c)-,c )>0$. But the equation in \eqref{w_pb} implies
$$\frac{\si^2}{2} w_{xx}(\sw(c)-,c )=\Big(2 u_x -(\mu-c_j) w_x+r w\Big)(\sw(c),c )=0,$$
leading to a contradiction. So \eqref{wx>0} holds, from which and $w(\sw(c),c )=0$ we conclude
\begin{align}\label{w<0}
w(x,c )<0,~ x\in(\sw(c)-\ep, \sw(c)).
\end{align}

Define
$$x^*=\sup\Big\{x\in [0,\sw(c)]\;\Big|\; w(x,c )> 0\Big\},$$
then \eqref{w<0} implies $x^*<\sw(c)$.
If $x^*>0$, then by the continuity of $w$, we have $w(x^*,c )=0$; otherwise, we have $x^*=0$ so that $w(x^*,c )=w(0,c )=0$. Integrating the equation in \eqref{w_pb} in $[x^*, \sw(c)]$, and recalling $w(\sw(c),c)=u_x(\sw(c),c)=0$, we have
$$\frac{\si^2}{2} w_{x}(x^*,c )=-r \int_{x^*}^{\sw(c)}w (x,c )\d x+2u (x^*,c ) >0.$$
It hence follows that $w(x^*+\ep',c)>w(x^*,c )=0$ for any sufficiently small $\ep'>0$. But this contradicts to the definition of $x^*$, completing the proof.
\end{proof}

Combining \eqref{C1}, \lemref{lem:xms}, \lemref{lem:x'c} and \propref{smooth}, \propref{lip} follows.

\section*{Acknowledgments} We thank Dr. Jiacheng Fan for his help on performing the numerical analysis. We also thank the associate editor and two anonymous referees for their valuable comments and suggestions that lead to a better version of this paper. 

\bibliographystyle{siamplain}
\bibliography{references}

\begin{thebibliography}{}


\bibitem{albrecher2020optimal}
Albrecher, H., Azcue, P., and Muler, N. (2020).
\newblock Optimal ratcheting of dividends in insurance.
\newblock {\em SIAM Journal on Control and Optimization}, 58(4):1822--1845.

\bibitem{albrecher2022optimal}
Albrecher, H., Azcue, P., and Muler, N. (2022).
\newblock Optimal ratcheting of dividends in a Brownian risk model.
\newblock {\em SIAM Journal on Financial Mathematics}, 13(3):657--701.

\bibitem{albrecher2018dividends}
Albrecher, H., B{\"a}uerle, N., and Bladt, M. (2018).
\newblock Dividends: From refracting to ratcheting.
\newblock {\em Insurance: Mathematics and Economics}, 83:47--58.

{
\bibitem{angoshtari2019optimal}
Angoshtari, B., Bayraktar, E., and Young, V.~R. (2019).
\newblock Optimal dividend distribution under drawdown and ratcheting constraints on dividend rates.
\newblock {\em SIAM Journal on Financial Mathematics}, 10(2):547--577.
}

\bibitem{angoshtari2022optimal}
Angoshtari, B., Bayraktar, E., and Young, V.~R. (2022).
\newblock Optimal investment and consumption under a habit-formation constraint.
\newblock {\em SIAM Journal on Financial Mathematics}, 13(1):321--352.

\bibitem{arun2012merton}
Arun, T. (2012).
\newblock The merton problem with a drawdown constraint on consumption.
\newblock {\em arXiv preprint arXiv:1210.5205}.

\bibitem{asmussen1997controlled}
Asmussen, S. and Taksar, M. (1997).
\newblock Controlled diffusion models for optimal dividend pay-out.
\newblock {\em Insurance: Mathematics and Economics}, 20(1):1--15.

\bibitem{avanzi2009strategies}
Avanzi, B. (2009).
\newblock Strategies for dividend distribution: A review.
\newblock {\em North American Actuarial Journal}, 13(2):217--251.

\bibitem{azcue2005optimal}
Azcue, P. and Muler, N. (2005).
\newblock Optimal reinsurance and dividend distribution policies in the
cram{\'e}r-lundberg model.
\newblock {\em Mathematical Finance: An International Journal of Mathematics,
Statistics and Financial Economics}, 15(2):261--308.

{
\bibitem{BayraktarEgami2010}
Bayraktar, E., Egami, M. (2010).
\newblock A unified treatment of dividend payment problems under fixed cost and implementation delays. 
\newblock {\em Math Meth Oper Res}, 71, 325-–351.
}

{
\bibitem{BayraktarKyprianou2013}
Bayraktar E, Kyprianou AE, Yamazaki K. (2013).
\newblock 
On optimal dividends in the dual model. 
\newblock {\em ASTIN Bulletin.}, 43(3), 359--372.
}

{
\bibitem{BayraktarYoung2008}
Bayraktar, E., Young, V. R. (2008).
\newblock 
Minimizing the probability of ruin when consumption is ratcheted.
\newblock {\em North American Actuarial Journal}, 12(4), 428--442. 
} 




\bibitem{belhaj2010optimal}
Belhaj, M. (2010).
\newblock Optimal dividend payments when cash reserves follow a jump-diffusion
process.
\newblock {\em Mathematical Finance: An International Journal of Mathematics,
Statistics and Financial Economics}, 20(2):313--325.

\bibitem{CLLL15}
Chen, X., Landriault, D., Li, B. and Li, D (2015).
\newblock On minimizing drawdown risks of lifetime investments
\newblock {\em Insurance: Mathematics and Economics}, 65:46--54.


\bibitem{DXZ10}
Dai M, Xu ZQ, Zhou XY (2010).
\newblock Continuous-time Markowitz's model with transaction costs.
\newblock {\em SIAM J. Financial Math.}, 1(1):96--125.

\bibitem{DY09}
Dai M, and Yi FH (2009).
\newblock Finite-horizon optimal investment with transaction costs: A parabolic double obstacle problem,
\newblock {\em J. Differential Equations}, 246:1445–1469.

\bibitem{de1957impostazione}
De~Finetti, B. (1957).
\newblock Su un’impostazione alternativa della teoria collettiva del rischio.
\newblock In {\em Transactions of the XVth international congress of
Actuaries}, volume~2, pages 433--443. New York.

\bibitem{dybvig1995dusenberry}
Dybvig, P.~H. (1995).
\newblock Dusenberry's ratcheting of consumption: optimal dynamic consumption
and investment given intolerance for any decline in standard of living.
\newblock {\em The Review of Economic Studies}, 62(2):287--313.

\bibitem{elie2008optimal}
Elie, R. and Touzi, N. (2008).
\newblock Optimal lifetime consumption and investment under a drawdown
constraint.
\newblock {\em Finance and Stochastics}, 12:299--330.

\bibitem{Ev17}
L.C. Evans.
{\em Partial Differential Equations.}
AMS, 2017.



\bibitem{Fr75}
A. Friedman.
{\em Parabolic variational inequalities in one space dimension and smoothness of the free boundary.}
Journal of Functional Analysis,
18:151-176, 1975.


\bibitem{gerber1969entscheidungskriterien}
Gerber, H.~U. (1969).
\newblock {\em Entscheidungskriterien f{\"u}r den zusammengesetzten
Poisson-Prozess}.
\newblock PhD thesis, ETH Zurich.

\bibitem{gerber2004optimal}
Gerber, H.~U. and Shiu, E.~S. (2004).
\newblock Optimal dividends: analysis with brownian motion.
\newblock {\em North American Actuarial Journal}, 8(1):1--20.

\bibitem{gerber2006optimal}
Gerber, H.~U. and Shiu, E.~S. (2006).
\newblock On optimal dividend strategies in the compound poisson model.
\newblock {\em North American Actuarial Journal}, 10(2):76--93.


\bibitem{TG11}
D. Gilbarg and N.S. Trudinger, {\em Elliptic Partial Differential Equations of Second Order}, Springer, Berlin, Heidelberg, New York, 1997.


\bibitem{jeon2018portfolio}
Jeon, J., Koo, H.~K., and Shin, Y.~H. (2018).
\newblock Portfolio selection with consumption ratcheting.
\newblock {\em Journal of Economic Dynamics and Control}, 92:153--182.

\bibitem{jeon2022finite}
Jeon, J. and Oh, J. (2022).
\newblock Finite horizon portfolio selection problem with a drawdown constraint
on consumption.
\newblock {\em Journal of Mathematical Analysis and Applications},
506(1):125542.

\bibitem{jeon2020optimal}
Jeon, J. and Park, K. (2020).
\newblock Optimal retirement and portfolio selection with consumption
ratcheting.
\newblock {\em Mathematics and Financial Economics}, 14(3):353--397.

{
\bibitem{Kyprianou2012}
Kyprianou AE, Loeffen R, Pérez J-L. (2012).
\newblock
Optimal Control with Absolutely Continuous Strategies for Spectrally Negative Lévy Processes. \newblock {\em Journal of Applied Probability}, 49(1):150--166. 
}

\bibitem{M08}
X. R. Mao, {\em Stochastic Differential Equations and Applications (Second Edition)}, Woodhead Publishing, 2008.


{
\bibitem{reppen2020optimal}
A~Max Reppen, Jean-Charles Rochet, and H~Mete Soner.
\newblock Optimal dividend policies with random profitability.
\newblock {\em Mathematical Finance}, 30(1):228--259, 2020.
}

\bibitem{roche2006optimal}
Roche, H. (2006).
\newblock Optimal consumption and investment strategies under wealth
ratcheting.
\newblock {\em preprint}.



\bibitem{Ta00}
M. Taksar.
{\em Optimal risk and dividend distribution control models for an
insurance company.}
Mathematical Methods of Operations Research, 51:1-42, 2000.


\bibitem{tian2020optimal}
Tian, L., Bai, L., and Guo, J. (2020).
\newblock Optimal singular dividend problem under the sparre andersen model.
\newblock {\em Journal of Optimization Theory and Applications}, 184:603--626.

\bibitem{YZ99}
J. Yong, and X. Zhou.
{\em Stochastic controls: Hamiltonian systems and HJB equations.}
Applications of Mathematics (New York) 43,
Springer-Verlag, New York, 1999.



\end{thebibliography}

\end{document}